\documentclass[submission]{fundam}

\newenvironment{prOOf}[1]
  {\trivlist\PRstyle\item[]{\bfseries Proof {#1}:}\newline}{\QED\endtrivlist}
\def\squareforqed{\hbox{\rlap{$\sqcap$}$\sqcup$}}
\def\QED{\ifmmode\squareforqed\else{\unskip\nobreak\hfil
\penalty50\hskip1em\null\nobreak\hfil\squareforqed
\parfillskip=0pt\finalhyphendemerits=0\endgraf}\fi}

\usepackage{graphicx}
\usepackage{amssymb,amsmath,wasysym}
\usepackage[mathscr]{euscript}
\usepackage[utf8]{inputenc}
\usepackage{hyperref}
\usepackage{multicol}
\usepackage{color,booktabs}
\title{Stochastic Cellular Automata:\\ Correlations, Decidability and Simulations}
\runninghead{P. Arrighi, N. Schabanel, G. Theyssier}{Stochastic Cellular Automata: Correlations, Decidability and Simulations}

\author{	Pablo Arrighi\thanks{This work has been partially funded by the ANR-10-JCJC-0208 CausaQ grant.}\\
			{Université de Grenoble (LIG, UMR 5217), France}\\[-1mm]
			{Université de Lyon (LIP, UMR 5668), France}
	\and 	Nicolas Schabanel\thanks{This work has been partially funded by the ANR-2010-BLAN-0204 Magnum and ANR-12-BS02-005 RDAM grants.}\\
			{CNRS, Université Paris Diderot (LIAFA, UMR 7089), France}\\[-1mm]
        		{Université de Lyon (IXXI), France}
 	\and 	Guillaume Theyssier\\
			{CNRS, Université de Savoie (LAMA, UMR 5127), France}
      }

\newcommand\DELETE[1]{}

\newcommand{\Qi}{\makebox{$\bigcirc\!\!\!\!\!\rightarrow$}}
\newcommand{\Qo}{\rightarrow}
\newcommand{\Qf}{\checked}
\newcommand{\Qe}{\bot}
\newcommand{\PPT}{\ensuremath{\textsc{PPT}}}

\newcommand{\PF}{\ensuremath{\operatorname{\mathscr{PF}}}}
\newcommand{\ZZ}{\ensuremath{\mathbb Z}}
\newcommand\Z\ZZ
\newcommand{\NN}{\ensuremath{\mathbb N}}
\newcommand\N\NN
\newcommand{\QQ}{\ensuremath{\mathbb Q}}

\newcommand{\AUTO}[1]{{\ensuremath{\mathcal{#1}}}}

\newcommand\CAA{\AUTO A}
\newcommand\CAB{\AUTO B}

\newcommand\Mes[1]{\ensuremath{\mathcal{M}({#1})}}

\newcommand\Pof[1]{\ensuremath{\mathcal{P}({#1})}}
\newcommand\DET[1]{\ensuremath{\mathcal{D}_{#1}}}
\newcommand\NDET[1]{\ensuremath{\mathcal{N}_{#1}}}
\newcommand\STOC[1]{\ensuremath{\mathcal{S}_{#1}}}
\newcommand\cyl[2]{\ensuremath{[{#1}]_{#2}}}

\newcommand{\SUB}{\sqsubseteq}
\newcommand{\PROJ}{\unlhd}
\newcommand{\MIX}{{\PROJ{}\hskip-5pt\SUB}}
\newcommand{\NSUB}{\overset{N}{\SUB}}
\newcommand{\NPROJ}{\overset{N}{\PROJ}}
\newcommand{\NMIX}{\overset{N}{\MIX}}
\newcommand{\SSUB}{\overset{S}{\SUB}}
\newcommand{\SPROJ}{\overset{S}{\PROJ}}
\newcommand{\SMIX}{\overset{S}{\MIX}}
\newcommand{\DSUB}{\overset{D}{\SUB}}
\newcommand{\DPROJ}{\overset{D}{\PROJ}}
\newcommand{\DMIX}{\overset{D}{\MIX}}
\newcommand\rest[2]{{}_{#1}{#2}}
\newcommand\proj[2]{{}^{#1\!}{#2}}
\newcommand\mix[3]{{{}^{#1}_{#2}}{#3}}

\newcommand\shift[1]{\mathfrak{\sigma}_{#1}}
\newcommand\grp[2]{{#1}^{\langle#2\rangle}}
\newcommand\bloc[1]{b_{#1}}
\newcommand\debloc[1]{b^{-1}_{#1}}

\newcommand\simu{\preccurlyeq}
\newcommand\ssimui{{\simu_i^S}}
\newcommand\ssimus{{\simu_\pi^S}}
\newcommand\ssimum{{\simu_m^S}}
\newcommand\nsimui{{\simu_i^N}}
\newcommand\nsimus{{\simu_\pi^N}}
\newcommand\nsimum{{\simu_m^N}}
\newcommand\dsimui{{\simu_i^D}}
\newcommand\dsimus{{\simu_\pi^D}}
\newcommand\dsimum{{\simu_m^D}}

\newcommand\somerel\leq

\newcommand\CFCA{\textsf{CFCA}}

 \renewcommand{\leq}{\leqslant}
 \renewcommand{\geq}{\geqslant}
 \renewcommand{\emptyset}{\varnothing}
\newcommand\probfa[1]{\mathbb{P}_{\mathcal{A}}\left(#1\right)}
 
 \newcommand{\anevent}{\operatorname{\mathcal E}}

\begin{document}
\maketitle

\begin{abstract}
This paper introduces a simple formalism for dealing with deterministic, non-deterministic and stochastic cellular automata in an unified and composable manner. This formalism allows for local probabilistic correlations, a feature which is not present in usual definitions. We show that this feature allows for strictly more behaviors (for instance, number conserving stochastic cellular automata require these local probabilistic correlations). We also show that several problems which are deceptively simple in the usual definitions, become undecidable when we allow for local probabilistic correlations, even in dimension one. Armed with this formalism, we extend the notion of intrinsic simulation between deterministic cellular automata, to the non-deterministic and stochastic settings. Although the intrinsic simulation relation is shown to become undecidable in dimension two and higher, we provide explicit tools to prove or disprove the existence of such a simulation between any two given stochastic cellular automata. Those tools rely upon a characterization of equality of stochastic global maps, shown to be equivalent to the existence of a stochastic coupling between the random sources. We apply them to prove that there is no universal stochastic cellular automaton. Yet we provide stochastic cellular automata achieving optimal partial universality, as well as a universal non-deterministic cellular automaton.
\end{abstract}

\section{Introduction} 

\paragraph{A motivation: stochastic simulation.} Cellular Automata (CA) are a key tool in simulating natural phenomena. This is because they constitute a privileged mathematical framework in which to cast the simulated phenomena, and they describe a massively parallel architecture in which to implement the simulator. 
Often however, the system that needs to be simulated is a noisy system. More embarrassingly even, it may happen that the system that is used as a simulator is again a noisy system. The latter is uncommon if one thinks of a classical computer as the simulator, but quite common for instance if one thinks of using a small scale model of a system as a simulator for that system.\\
Fortunately, when both the simulated system and the simulating system are noisy, it may happen that both effects cancel out, i.e. that the noise of the simulator is made to coincide with that of the simulated. In such a situation a model of noise is used to simulate another, and the simulation may even turn out to be\ldots exact. This paper begins to give a formal answer to the question: When can it be said that a noisy system is able to exactly simulate another?\\
This precise question has become crucial in the field of quantum simulation. Indeed, there are many quantum phenomena which we need to simulate, and these in general are quite noisy. Moreover, only quantum computers are able to simulate them efficiently, but in the current state of experimental physics these are also quite noisy. Could it be that noisy quantum computers may serve to simulate a noisy quantum systems? The same remark applies to Natural Computing in general. Still, the question is challenging enough in the classical setting.

\paragraph{A challenge: the need for local probabilistic correlations.} The first problem that one comes across is that stochastic CA have only received little attention from the theoretical community. When they have been considered, they were usually defined as the application of a probabilistic function uniformly across space \cite{Toom,Gacs,Fates,Mairesse,RST,FatesRST06}. In this paper we will refer to this model as local Correlation-Free CA (CFCA). Indeed, this particular class of stochastic CA has the unique property that, starting from a determined configuration, the cell's distributions remain uncorrelated after one step. This was pointed out in \cite{DBLP:conf/cie/ArrighiFNT11}, which provides an example (cf. $\textsc{Parity}$ stochastic CA which we will use later) which cannot be realized as CFCA, in spite of the fact that they require only local probabilistic correlations and hence fit naturally in the CA framework. Moreover, \cite{DBLP:conf/cie/ArrighiFNT11} points out that the composition of two CFCA is not always a CFCA. The lack of composablity of a model is an obstacle for defining intrinsic simulation, because the notion must be defined up to grouping in space and in time. In \cite{DBLP:conf/cie/ArrighiFNT11} a composable model is suggested, but it lacks formalization.\\ 
In this paper we propose a simple formalism to deal with general stochastic CA. The formalism relies on considering a CA $F(c,s)$ fed, besides the current configuration $c$, with a new fresh independent uniform random configuration $s$ at every time step. This allows any kind of local probabilistic correlations and includes in particular all the examples of \cite{DBLP:conf/cie/ArrighiFNT11}. As it turns out, the definition also captures deterministic and non-deterministic CA (non-deterministic CA are obtained by ignoring the probability distribution over the random configuration). 

\paragraph{Results on stochastic simulation.} This formalism allows us to extend the notions of simulation developed for the deterministic setting \cite{bulking1,bulking2}, to the non-deterministic and stochastic settings. The choice of making explicit the random source in the formalism has turned out to be crucial to tackle the second problem, as it allows a precise analysis of the influence of randomness, in terms of simulation power. \\
The second problem that one comes across is that the question of whether two such stochastic CA are equal in terms of probability distributions is highly non-trivial. In particular, we show that testing if two stochastic CA define the same random map becomes undecidable in dimension $2$ and higher (Theorem~\ref{thm:deciding}). Still, we provide an explicit tool (the coupling of the random sources of two stochastic CA) that allows to prove (or disprove) the equality of their probability distributions. More precisely, we show that the existence of such a coupling is strictly equivalent to the equality of the distribution of the random maps of two stochastic CA (Theorem~\ref{thm:coupling}).\\
The choice of making explicit the random source allows us to show some no-go results. Any stochastic CA may only simulate stochastic CA with a compatible random source (where compatibility is expressed as a simple arithmetic equation, Theorem~\ref{thm:primefactors}). It follows that there is no universal stochastic CA (Corollary~\ref{cor:no:stoc:universal}). Still, we show that there is a universal CA for the non-deterministic dynamics (Theorem~\ref{thm:nondet:univ}), and we are able to provide a universal stochastic CA for every class of compatible random source (Theorem~\ref{thm:stoc:univ}).  

\paragraph{Results on questions of computability versus allowing for local probabilistic correlations.} 
The fact that testing if two stochastic CA define the same random map is undecidable in dimension $2$ and higher, which is not the case in the particular case of CFCA, suggested that many problems appear deceptively simple in the CFCA formalism. And indeed, their difficulty comes back as soon as one iterates the CFCA: for instance, we show that in dimension $2$ and higher, it is undecidable whether the squares of two CFCA define the same random map (Corollary \ref{cor:undeciteratedequalityforCFCA}). Worse even, it is undecidable whether the square of a CFCA is noisy (Theorem \ref{thm:undecnoisecf}). \\
In dimension one, these problems, and many others, are shown to be decidable for general stochastic CA (Corollary~\ref{cor:dec1Dequality} and \ref{cor:buchi}). Thus, they cannot serve to point out a separation with CFCA. Yet, we show that the Pattern-Probability-Threshold problem (i.e. the question whether some pattern can appear with a probability higher that some threshold) is again undecidable for stochastic CA (Theorem \ref{thm:pptundeci}), whereas it is decidable for CFCA (Theorem \ref{th:pptdeccfca}). We also show that some behaviors are out-of-reach of CFCA. Indeed, CFCA cannot be number-conserving unless they are deterministic (Lemma \ref{lem:CFCA:num:det}). Moreover, even the iterates of a CFCA cannot be surjective number-conserving unless they are deterministic (Theorem \ref{thm:CFCA:num:det}). Iterates of two-state CFCA cannot reproduce behaviors alike the $\textsc{Parity}$ example either (Theorem \ref{th:reven}).

\paragraph{Plan.} Section~\ref{sec:basic} recalls the vital minimum about probability theory. Section~\ref{sec:SCA} states our formalism. Section~\ref{sec:localglobal} explains that some global properties of the SCA (such as equality of stochastic global functions, the probability of appearance of certain patterns, injectivity and surjectivity) cannot be decided from the local rules, but may be approached through some other proof techniques. These problems tend to simplify for the particular case of local Correlation-Free SCA, which corresponds to the more usual definition, and for the one-dimensional case, as explained in Sections~\ref{sec:correlationfree} and \ref{sec:dimensionone}. This shows that local Correlation-Free SCA are fundamentally simpler: some behaviors cannot be reached. Section~\ref{sec:simul} extends the notion of  intrinsic simulation to the non-deterministic and stochastic settings.  Section~\ref{sec:univ} provides the no-go results in the stochastic setting, the universality constructions. Section~\ref{sec:open} concludes this article with a list of open questions.

\section{Standard Definitions}
\label{sec:basic}

\textit{Even if this article focuses mainly on one-dimensional CA for the sake of simplicity, it extends naturally to higher dimensions. Each time a result is sensitive to dimension, it will be explicited in the statement.}
\newcommand\extendstoalldimensions{\textit{The proof is written for 1D CA to simplify notations but it extends to any dimension in a straighforward way.}}

For any finite set $A$ we consider the symbolic space $A^\ZZ$. For any $c\in A^\ZZ$ and $z\in\ZZ$ we denote by $c_z$ the value of $c$ at point $z$. $A^\ZZ$ is endowed with the Cantor topology (infinite product of the discrete topology on each copy of $A$) which is compact and metric (see \cite{kurkabook} for details). A basis of this topology is given by cylinders which are actually clopen sets: given some finite word $u$ and some position $z$, the cylinder $\cyl{u}{z}$ is the set
${\cyl{u}{z} = \{c\in A^\ZZ: \forall x, 0\leq x<|u|-1, c_{z+x}=u_x\}.}$

We denote by $\Mes{A^\ZZ}$ the set of Borel probability measures on $A^\ZZ$. By Carath\'eodory extension theorem, Borel probability measures are characterized by their value on cylinders. Concretely, a measure is given by a function $\mu$ from cylinders to the real interval $[0,1]$ such that $\mu(A^\ZZ) = 1$ and
\[\forall u\in A^*, \forall z\in\ZZ, \quad \mu(\cyl{u}{z}) = \sum_{a\in A}\mu(\cyl{ua}{z}) = \sum_{a\in A}\mu(\cyl{au}{z-1})\]

We denote by $\nu_A$ the uniform measure over $A^\ZZ$ (s.t. $\nu_A(\cyl{u}{z})  = \frac{1}{|A|^{|u|}}$). We shall denote it as $\nu$ when the underlying alphabet~$A$ is clear from the context. 

We endow the set $\Mes{A^\ZZ}$ with the compact topology given by the following distance:
${\mathfrak{D}(\mu_1,\mu_2) = \sum_{n\geq 0}2^{-n}\cdot\max_{u\in A^{2n+1}} \bigl|\mu_1(\cyl{u}{-n})-\mu_2(\cyl{u}{-n})\bigr|}$.
See \cite{Pivato09} for a review of works on cellular automata from the measure-theoretic point of view.

\section{Stochastic Cellular Automata}
\label{sec:SCA}

Non-deterministic and stochastic cellular automata are captured by the same syntactical object given in the following definition. They differ only by the way we look at the associated global behavior. Moreover, deterministic CA are a particular case of stochastic CA and can also be defined in the same formalism.

\subsection{The Syntactical Object}

\begin{definition}
\label{def:syntax}
A \emph{stochastic cellular automaton} $\AUTO A=(Q,R,V,V',f)$ consists in:
\begin{itemize}
\item a finite set of \emph{states} $Q$
\item a finite set $R$ called the \emph{random symbols}
\item two finite subsets of $\mathbb Z$: $V=\{v_1,\ldots, v_\rho\}$ and $V'=\{v'_1,\ldots, v'_{\rho'}\}$, called the \emph{neighborhoods}; $\rho$ and $\rho'$ are the sizes of the neighborhoods and ${k = \max_{v\in V\cup V'}|v|}$ is the \emph{radius} of the neighborhoods.
\item a \emph{local transition function} $f: Q^{\rho} \times R^{\rho'} \rightarrow Q$
\end{itemize}
A function $c\in Q^\ZZ$ is called a \emph{configuration}; $c_j$ is called the \emph{state} of the \emph{cell}~$j$ in configuration~$c$. A function $s\in R^\ZZ$ is called an \emph{$R$-configuration}.

In the particular case where $V'=\{0\}$ (i.e., where each cell uses its own random symbol only), we say that \CAA{} is a \emph{Correlation-Free Cellular Automaton} (\CFCA{} for short).
\end{definition}

\begin{definition}[Explicit Global Function]
 To this local description, we associate the \emph{explicit global function} $F:Q^\ZZ\times R^\ZZ\rightarrow Q^\ZZ$ defined for any configuration $c$ and $R$-configuration $s$ by:
${F(c,s)_z = f\bigl((c_{z+v_1},\ldots,c_{z+v_\rho}),(s_{z+v'_1},\ldots,s_{z+v'_{\rho'}})\bigr).}$
Given a sequence $\bigl(s^t\bigr)_t$ of $R$-configurations and an initial configuration $c$, we define the associated \emph{space-time diagram} as the bi-infinite matrix ${\bigl(c^t_z\bigr)_{t\geq 0, z\in\ZZ}}$ where $c^t\in Q^\ZZ$ is defined by $c^0=c$ and ${c^{t+1}=F(c^t,s^t)}$. We also define for any $t\geq 1$ the $t^\text{th}$ iterate of the explicit global function $F^t:Q^\ZZ\times \bigl(R^\ZZ\bigr)^t\rightarrow Q^\ZZ$ by $F^0(c)=c$ for all configuration~$c$ and 
\[F^{t+1}(c,s^1,\ldots,s^{t+1}) = F\bigl(F^t(c,s^1,\ldots,s^t),s^{t+1}\bigr)\]
so that ${c^t=F^t(c,s^1,\ldots,s^t)}$.\DELETE{ We denote by $\CAA^t = (Q,R^t,V_+,V'_+,f^t)$ the stochatic CA whose explicit global function is $F^t$.}
\end{definition}

In this paper, we adopt the convention that local functions are denoted by a lowercase letter (typically~$f$) and explicit global functions by the corresponding capital letter (typically $F$). Moreover, we will often define CA through their explicit global function since details about neighborhoods often do not matter in this paper.

The explicit global function captures all possible actions of the automaton on configurations. This function allows to derive three kinds of dynamics: deterministic, non-deterministic and stochastic.

\subsection{Deterministic and Non-Deterministic Dynamics}

\paragraph{Deterministic.}  The \emph{deterministic global function} ${\DET{F}:Q^\ZZ\rightarrow Q^\ZZ}$ of ${\CAA=(Q,R,V,V',f)}$ is  defined by ${\DET{F}(c) = F(c,0^\ZZ)}$ where $0$ is a distinguished element of $R$. \CAA{}  is said to be \emph{deterministic} if its local transition function~$f$ does not depend on its second argument (the random symbols).

\paragraph{Non-Deterministic.} The \emph{non-deterministic global function} ${\NDET{F}:Q^\ZZ\rightarrow \Pof{Q^\ZZ}}$ of~$\CAA$ is  defined for any configuration $c\in Q^\ZZ$ by 
${\NDET{F}(c) = \{F(c,s):s\in R^\ZZ\}}$.

\paragraph{Dynamics.} The deterministic dynamics of \CAA{} is given by the sequence of iterates $(\DET{F}^t)_{t\geq 0}$. Similarly the non-deterministic dynamics of $\CAA$ is given by the iterates $\NDET{F}^t:Q^\ZZ\rightarrow \Pof{Q^\ZZ}$ defined by $\NDET{F}^{0}(c) = \{c\}$ and $\NDET{F}^{t+1}(c) = \bigcup_{c'\in\NDET{F}^t(c)}\NDET{F}(c')$.

\subsection{Stochastic Dynamics}
\label{sec:stochdyn}

 The stochastic point of view consists in taking the $R$-component as a source of randomness. More precisely, the explicit global function $F$ is fed at each time step with a
random uniform and independent $R$-configuration. This defines a stochastic process for which we are then interested in the distribution of states across space and time. By Carath\'eodory extension theorem, this distribution is fully determined by the probabilities of the events of the form  ``starting from $c$, the word $u$ occurs at position $z$ after $t$ steps of the process''.
Formally, for $t=1$, this event is the set:
\[\anevent_{c,\cyl{u}{z}}=\bigl\{s\in R^\ZZ : F(c,s)\in\cyl{u}{z}\bigr\}.\]
In order to evaluate the probability of this event, we use the locality of the explicit global function $F$. The event ``$F(c,s)\in\cyl{u}{z}$'' only depends of the cells of $s$ from position ${a=z-k}$ to position ${b=z+k+|u|-1}$. Therefore, if ${J = \{v\in R^{b-a}:F(c,[v]_a)\subseteq[u]_z\}}$, then $\anevent_{c,\cyl{u}{z}}=\cup_{v\in J} [v]_a$ and hence $\anevent_{c,\cyl{u}{z}}$ is a measurable set of probability: $\nu_R(\anevent_{c,\cyl{u}{z}})=\sum_{v\in J}\nu_R(\cyl{v}{a})=|J|/|R|^{b-a}$ (recall that $\nu_R$ is the uniform measure over $R^\ZZ$).


More generally to  any CA \CAA{} we associate its \emph{stochastic global function} ${\STOC{F}:Q^\ZZ\rightarrow\Mes{Q^\ZZ}}$ defined for any configuration $c\in Q^*$ by: $\forall u\in Q^\ZZ, \forall z\in\ZZ,$
 \[
\bigl(\STOC{F}(c)\bigr)(\cyl{u}{z}) = \nu_R(\anevent_{c,\cyl{u}{z}}) \text{ = the probability of event $\anevent_{c,\cyl{u}{z}}$.}
\]

\paragraph{Example.} For instance, consider the stochastic function $\textsc{Parity}$ that maps every configuration over the alphabet $\{0,1,\#\}$ to a random configuration in which every $\{0,1\}$-word delimited by two consecutive $\#$  is replaced by a random independent uniform word of length with even parity. This cannot be realized by a \CFCA{}. Still, one can realize the stochastic function $\textsc{Parity}$ as a stochastic CA with $Q=\{\#,0,1\}$, $R=\{0,1\}$ and local rule ${f:Q^{\{-1,0,1\}}\times R^{\{-1,0\}}\rightarrow Q}$ given by: for all $c_{-1},c_0,c_1,s_{-1},s_0\in\{0,1\}$  and $a,b\in\{\#,0,1\}$,
$$
\begin{array}{c@{\quad\quad}c@{\quad\quad}c}
	f(a\#b,s_{-1}s_0) = \#
&	f(\#c_0\#,s_{-1}s_0) = 0
&	f(\#c_0c_1,s_{-1}s_0) = s_0\\[2mm]
\multicolumn{3}{c}{
	f(c_{-1}c_0\#,s_{-1}s_0) = s_{-1}
\quad\quad	f(c_{-1}c_0c_1,s_{-1}s_0) = s_{-1}+s_0}
\end{array}
$$
One can easily check that this local probabilistic correlations ensures that every word delimited by two consecutive $\#$ is indeed mapped to a uniform independent random word of even parity.

\paragraph{Dynamics.} As opposed to the deterministic and non-deterministic setting, defining an iterate of this map is a not so trivial task.  There are two approaches: defining directly the measure after $t$ steps or extending the map $\STOC{F}$ to a map from $\Mes{Q^\ZZ}$ to itself. Both rely crucially on the continuity of $F$. In particular, we want to make sure that the definition of the measure after $t$ steps matches $t$ iterations of the one-step map, and hence, is independent of the explicit mechanics of $F$ but depends only on the map $\STOC{F}$ defined by $F$. 

The easiest one to present is the first approach. For any $t\geq 1$, the event $\anevent^t_{c,\cyl{u}{z}}$ that the word $u$ appears at position~$z$ at time~$t$ from configuration~$c$ consists in the set of all $t$-uples of random configurations $(s^1,\ldots,s^t)$ yielding $u$ at position~$z$ from~$c$, i.e.: 
\[\anevent^t_{c,\cyl{u}{z}}=\bigl\{(s^1,\ldots,s^t)\in\bigl(R^\ZZ\bigr)^t : F^t(c,s^1,\ldots,s^t)\in\cyl{u}{z}\bigr\}\]
As before $\anevent^t_{c,\cyl{u}{z}}$ is a measurable set in $\bigl(R^\ZZ\bigr)^t$ because it is a product of finite unions of cylinders by the locality of $F$. We therefore define $\STOC{F}^t:Q^\ZZ\rightarrow\Mes{Q^\ZZ} $, the iterate of the stochastic global function, by:
\[\bigl(\STOC{F}^t(c)\bigr)(\cyl{u}{z}) = \nu_{R^t}(\anevent^t_{c,\cyl{u}{z}})
\text{ =  the probability of event $\anevent^t_{c,\cyl{u}{z}}$}
\]
where $\nu_{R^t}$ denotes the uniform measure on the product space $\bigl(R^\ZZ\bigr)^t$. For similar reasons as above, $\STOC{F}^t(c)$ is a well-defined probability measure.

For all $t\geq 0$ and all words $u\in Q^n$ with $n\geq 2kt+1$, we will also denote by $F^t(u)$ the random variable for the random image $v\in Q^{n-2kt}$ of $u$ by $F^t$, defined formally as: for all $u\in Q^n$ and $v\in Q^{n-2kt}$, 
$$
\Pr\{ F^t(u) = v\} =  \bigl(\STOC{F}^t(c)\bigr)(\cyl{v}{kt}), \text{\quad for any $c\in \cyl{u}{0}$.} 
$$

The following key technical fact ensures that two automata define the same distribution over time as soon as their one-step distributions match.
\begin{fact}\label{fac:iterates}
  Let $\AUTO{A}$ and $\AUTO{B}$ be two stochastic CA with the same set of states $Q$ (and possibility different random alphabet) and of explicit global functions $F$ and $G$ respectively. If $\STOC{F}=\STOC{G}$ then for all $t\geq 1$ we have $\STOC{F}^t=\STOC{G}^t$
\end{fact}

\begin{proof}
  \extendstoalldimensions{}
  Consider a CA of explicit global function $F$. Consider a word $u$ and a position $z$. Let $\phi:Q^\ZZ\rightarrow \Pof{R^\ZZ}$ be function that associates to a configuration~$c$ the event $\anevent_{c,\cyl uz}$. $\phi(c)$ is entirely determined by the states of the cells from positions $a=z-k$ to $b=z+|u|+k$ in~$c$ (locality of $F$). Therefore $\phi$ is constant over every cylinder $\cyl{v}{a}$ with $v\in Q^{b-a}$. If we distinguish some $c_v\in\cyl va$ for every $v\in Q^{b-a}$, we obtain by definition of $F^t$ and  continuity of $F$:
\[
 \anevent^{t+1}_{c,\cyl{u}{z}} = \bigcup_{v\in Q^{b-a}} \left(\anevent^t_{c,\cyl{v}{a}}\times \anevent_{c_v,\cyl{u}{z}}\right).
 \]
Then, since sets $\bigl(\anevent^t_{c,\cyl{v}{a}}\bigr)_{v\in Q^{b-a}}$ are pairwise disjoint (because $F$ is deterministic and cylinders $\cyl{v}{a}$ are pairwise disjoint), we have
\begin{align*}
  \bigl(\STOC{F}^{t+1}(c)\bigr)(\cyl{u}{z}) =
  \nu_{R^{t+1}}(\anevent^{t+1}_{c,\cyl{u}{z}}) &=
  \sum_{v\in Q^{b-a}} \nu_{R^{t+1}}\bigl(\anevent^t_{c,\cyl{v}{a}}\times
  \anevent_{c_v,\cyl{u}{z}}\bigr)\\ &=
  \sum_{v\in Q^{b-a}}\nu_{R^t}(\anevent^t_{c,\cyl{v}{a}})\cdot\nu_R(\anevent_{c_v,\cyl{u}{z}})\\
  &= \sum_{v\in Q^{b-a}} \bigl(\STOC{F}^{t}(c)\bigr)(\cyl{v}{a})\cdot\bigl(\STOC{F}(c_v)\bigr)(\cyl{u}{z})
\end{align*}
The value of $S_F^t(c)$ over cylinders  can thus be expressed recursively as a function of a finite number of values $S_F$ over a finite number of cylinders. It follows that if for some pair of CA \CAA\/ and \CAB\/ with explicit global functions $F$ and $G$ we have  $\STOC{F}=\STOC{G}$, then $\STOC{F}^{t}=\STOC{G}^{t}$ for all $t$.
\end{proof}

In our setting one can recover the non-deterministic dynamics from the stochastic dynamics of a given stochastic CA. This heavily relies on the continuity of explicit global functions and compacity of symbolic spaces.

\begin{fact}\label{fac:stocndet}
  Given two CA with same set of states and explicit global functions $F_A$ and $F_B$, if $\STOC{F_A}=\STOC{F_B}$ then $\NDET{F_A}=\NDET{F_B}$.
\end{fact}
\begin{proof}
  Given some stochastic CA of explicit global function $F$, some configuration $c$ and some cylinder $\cyl{u}{z}$ we have 
\[\NDET{F}(c)\cap\cyl{u}{z}\neq\emptyset\Leftrightarrow \anevent_{c,\cyl{u}{z}}\neq\emptyset\Leftrightarrow\bigl(\STOC{F}(c)\bigr)(\cyl{u}{z})>0\]
by definition of $\NDET{F}$ and $\STOC{F}$. Since $\NDET{F}(c)$ is a closed set (continuity of $F$) it is determined by the set of cylinders intersecting it (compacity of the space). Hence $\NDET{F}(c)$ is determined by $\STOC{F}(c)$. The lemma follows. 
\end{proof}

\section{From Local to Global}
\label{sec:localglobal}

It is well-known in deterministic CA that determining global properties from the local representation is a generally hard problem. The purpose of this section is to show that the situation is even worse for stochastic CA.

\subsection{Equality of random maps: undecidability and explicit tools}
\label{sec:coupling}

\paragraph{An undecidable task for dimension~$2$ and higher.}
In the classical deterministic case, it is easy to determine whether two CA have the same global function. Equivalently determining whether two stochastic CA, as syntactical objects, have the same explicit global functions $F$ and $G$ is easy. However, given two stochastic CA which have possibly different explicit global function $F$ and $G$, it still may happen that $\NDET{F}=\NDET{G}$ or $\STOC{F}=\STOC{G}$, and determining whether this is the case turns out to be a difficult problem. In fact, Theorem~\ref{thm:deciding} states that these two decision problems are at least as difficult as the surjectivity problem of classical CA.

\begin{theorem}\label{thm:deciding}
  Let $\mathcal{P}_N$ (resp. $\mathcal{P}_S$) be the problem of deciding whether two given stochastic CA have the same non-deterministic (resp. stochastic) global function. The surjectivity problem of classical deterministic CA is reducible to both $\mathcal{P}_N$ and $\mathcal{P}_S$.
\end{theorem} 
\newcommand{\proofthdeciding}{
\begin{proof}
  \extendstoalldimensions{}
  Consider a classical CA $F:Q^{\ZZ}\rightarrow Q^{\ZZ}$ and define $\mu_F$ as the image by $F$ of the uniform measure $\mu_0$ on $Q^{\ZZ}$:
  \[\mu_F(\cyl{u}{z}) = \nu\bigl(F^{-1}(\cyl{u}{z})\bigr)\]
  It is well-known that $F$ is surjective if and only if ${\mu_F=\nu}$ (this result is true in any dimension; the proof for dimension 1 is in \cite{Hedlund} and follows from \cite{maki} for higher dimensions, but we recommend \cite{Pivato09} for a modern exposition in any dimension).
  
  Now let us define the stochastic CA $\CAA = (Q,Q,V,V',g)$ such that ${G(c,s) = F(s)}$. With this definition, \CAA{} is such that, for all $c$, ${\STOC{G}(c) = \mu_F}$. Hence, ${\STOC{G}(c)}$ is the uniform measure for any $c$ if and only if $F$ is surjective. We have also ${\NDET{G}(c)=Q^\ZZ}$ for all $c$ if and only if $G$ is surjective. The theorem follows since \CAA{} is recursively defined from $F$. 
\end{proof}
}\proofthdeciding

Surjectivity of classical CA is an undecidable property in dimension 2 and higher \cite{Kari94}.  As an immediate corollary, we get undecidability of equality of global maps in dimension $2$ and higher.

\begin{corollary}
  Fix ${d\geq 2}$. Problems $\mathcal{P}_N$ and $\mathcal{P}_S$ are undecidable for CA of dimension $d$.
\end{corollary}

\paragraph{Explicit tools for (dis)proving equality.}
Even if testing the equality of the non-deterministic or stochastic dynamics of two stochastic CA is undecidable for dimension $2$ and higher, Theorem~\ref{thm:coupling} states that equality, when it holds, can always be certified in terms of a \emph{stochastic coupling}. Indeed the stochastic coupling, by matching their two source of randomness, serves as a witness of the equality of the stochastic CA. This provides us with a very useful technique, because the existence of such a coupling is easy to prove or disprove in many concrete examples. Again the result heavily relies on the continuity of the explicit global function~$F$.

Let us first recall the standard notion of coupling.
\begin{definition}
  Let $\mu_1\in\Mes{Q_1^\ZZ}$ and $\mu_2\in\Mes{Q_2^\ZZ}$. A coupling of $\mu_1$ and $\mu_2$ is a measure $\gamma\in \Mes{Q_1^\ZZ\times Q_2^\ZZ}$ such that for any measurable sets $\anevent_1$ and $\anevent_2$, ${\gamma(\anevent_1\times Q_2^\ZZ)=\mu_1(\anevent_1)}$ and ${\gamma( Q_1^\ZZ\times\anevent_2)=\mu_2(\anevent_2)}$.
\end{definition}
Concretely, a coupling couples two measures so that each is recovered when the other is ignored. The motivation in defining a coupling is to bind the two distributions in order to prove that they induce the same kind of behavior: for instance, one can easily couple the two uniform measures over $\{1,2\}$ and $\{1,2,3,4\}$ so that with probability~$1$, both numbers will have the same parity ($\gamma$ gives a probability $1/4$ to each pair $(1,1)$, $(2,2)$, $(1,3)$ and $(2,4)$ and $0$ to the others). This demonstrates that the parity function is identically distributed in both cases.  

Theorem~\ref{thm:coupling} states that the dynamics of two stochastic CA are identical if and only if there is a  coupling of their random configurations so that their stochastic global functions become almost surely identical. This is one of our main results. 
\begin{definition}
  Two stochastic cellular automata, $\CAA_1=(Q,R_1,V_1,V_1',f_1)$ and $\CAA_2=(Q,R_2,V_2,V_2',f_2)$, with the same set of states~$Q$ are \emph{coupled} on configuration $c\in Q^\ZZ$ by a measure $\gamma\in\Mes{R_1^\ZZ\times R_2^\ZZ}$ if 
  \begin{enumerate}
  \item $\gamma$ is a coupling of the uniform measures on $R_1^\ZZ$ and $R_2^\ZZ$;
  \item $\gamma\bigl(\{(s_1,s_2)\in R_1^\ZZ\times R_2^\ZZ : F_1(c,s_1)=F_2(c,s_2)\}\bigr) = 1$, i.e. $F_1$ and $F_2$ produce almost surely the same image when fed with the $\gamma$-coupled random sources.
  \end{enumerate}
  
\end{definition}
Note that the set of pairs $(s_1,s_2)$ defined above is measurable because it is closed ($F_1$ and $F_2$ are continuous).

\begin{theorem}\label{thm:coupling}
  Two stochastic CA with the same set of states have the same stochastic global function if and only if, on each configuration $c$, they are coupled by some measure $\gamma_c$ (which depends on~$c$).
\end{theorem}
\noindent {\em Outline of the proof.} We fix a configuration $c$. By continuity of the explicit global functions, we construct a sequence of partial couplings $(\gamma_c^n)$ matching the random configurations of finite support of radius $n$. We then extract the coupling $\gamma_c$ from $(\gamma_c^n)$ by compacity of $\Mes{R_A^\ZZ\times R_B^\ZZ}$.

{
\begin{proof}
\extendstoalldimensions{}
 First, if $\CAA_1=(Q,R_1,V_1,V_1',f_1)$ and $\CAA_2=(Q,R_2,V_2,V_2',f_2)$ are coupled by $\gamma_c$ on configuration $c\in Q^\ZZ$, consider for any cylinder $\cyl{u}{z}$ the sets
  \begin{align*}
    \anevent_1 &= \{s\in R_1^\ZZ : F_1(c,s)\in\cyl{u}{z}\}\\
    \anevent_2 &= \{s\in R_2^\ZZ : F_2(c,s)\in\cyl{u}{z}\}\\
    X &= \{(s_1,s_2) \in R_1^\ZZ\times R_2^\ZZ : F_1(c,s_1)=F_2(c,s_2)\}
  \end{align*}
  Then, by the property of the coupling by $\gamma_c$,  we have 
\[\bigl(\STOC{F_1}(c)\bigr)(\cyl{u}{z}) = \nu_1(\anevent_1) = \gamma_c(\anevent_1\times R_2^\ZZ)\]
 where $\nu_1$ is the uniform measure on $R_1^\ZZ$. But $\gamma_c(\anevent_1\times R_2^\ZZ) = \gamma_c\bigl((\anevent_1\times R_2^\ZZ) \cap X\bigr)$ since $\gamma_c(X)=1$. Symmetrically we have \[\bigl(\STOC{F_2}(c)\bigr)(\cyl{u}{z}) = \gamma_c\bigl((R_1^\ZZ\times \anevent_2) \cap X\bigr).\]
 But, by definition of sets $\anevent_1$, $\anevent_2$ and $X$, we have ${R_1^\ZZ\times \anevent_2 \cap X = \anevent_1\times R_2^\ZZ \cap X}$. We conclude that $\STOC{F_1}=\STOC{F_2}$.
 For the other direction of the theorem, suppose $\STOC{F_1}=\STOC{F_2}$ and fix some configuration $c$. We denote by $\mu$ the measure ${\STOC{F_1}(c)=\STOC{F_2}(c)}$. Without loss of generality we can suppose that $\CAA_1$ and $\CAA_2$ have same radii $r$: $r_1=r_2=r$. We construct a sequence $(\gamma^n)$ of measures from which we can extract a limit point (by compacity of the space of measures) which is a valid coupling of $\CAA_1$ and $\CAA_2$ on configuration $c$.
\newcommand\ccyl[1]{[{#1}]}
To simplify the proof we focus on centered cylinders: for any word $w$ of odd length, we denote by $\ccyl{w}=\cyl{w}{z_w}$ where ${z_w = -\frac{|w|-1}{2}}$.
Let's fix $n$. For any word $u\in Q^{2n+1}$ we define:
\begin{align*}
  S_u^1 &= \{s\in R_1^\ZZ : F_1(c,s)\in\ccyl{u}\}\\
  S_u^2 &= \{s\in R_2^\ZZ : F_2(c,s)\in\ccyl{u}\}
\end{align*}
$F_1$ and $F_2$ being of radius $k$ we can write $S_u^i$ as a finite union of centered cylinders of length ${2(n+k)+1}$:
\[S_u^i = \bigcup_{v\in P_u^i}\ccyl{v}\]
where ${P_u^i\subseteq R_i^{2(n+k)+1}}$. Define the following partition $I_u^i$ of the real interval $[0,1)$ by:
\newcommand\ord[1]{\mathsf{rank}(#1)}
\[I_u^i(v) = \left[\frac{\ord{v}}{\# P_u^i};\frac{\ord{v}+1}{\# P_u^i}\right[\]
where $\ord{v}\in\{0,\ldots,\# P_u ^i - 1\}$ is the rank of $v$ in some arbitrarily chosen total ordering of $P_u^i$ (the lexicographical order for instance).
Since the sets $P_u^i$ form a partition of ${R_i^{2(n+k)+1}}$ when $u$ ranges over all words of $Q^{2n+1}$, we have for any $v\in R_i^{2(n+k)+1}$:
\[|I_u^i(v)| = \frac{1}{\# P_u^i} = \frac{\nu_i(\ccyl{v})}{\mu(\ccyl{u})}\]
(recall that $\nu_i$ stands for the uniform measure over $R_i^\ZZ$).
Now, for every ${v^1\in R_1^{2(n+k)+1}}$ and ${v^2\in R_2^{2(n+k)+1}}$, we construct $\gamma^n$ as:
\[\gamma^n(\ccyl{v^1},\ccyl{v^2}) =
\begin{cases}
  |I_u^1(v^1)\cap I_u^2(v^2)|\cdot\mu_Q(\ccyl{u}) &\text{ if $\exists u$ s.t. ${v^i\in P_u^i}$ for both $i=1,2$}\\
  0 &\text{ otherwise.}
\end{cases}\]
Furthermore, if $0_i$ is a distinguished element of $R_i$, we extend the definition of $\gamma^n$ to any pair $v^1\in R_1^{2m+1}$ and $v^2\in R_2^{2m+1}$ with ${m\geq n+k}$ by:
\[\gamma^n(\ccyl{v^1},\ccyl{v^2}) =
\begin{cases}
  \gamma^n(\ccyl{w^1},\ccyl{w^2}) &\text{ if $v^i = 0_i^{m-n-k}w^i0_i^{m-n-k}$ for $i=1,2$}\\
  0 &\text{ else.}
\end{cases}
\]
By $\sigma$-additivity $\gamma^n$ is thus defined on any cylinder and by extension theorem is a well-defined measure. Now by construction, we have for any $v^1\in R^{2(n+k)+1}$:
\[\gamma^n(\ccyl{v^1},R_2^\ZZ) = |I_u^1(v^1)|\cdot\mu_Q(\ccyl{u})\]
for some $u$ such that ${v^1\in P_u^1}$. Hence, ${\gamma^n(\ccyl{v^1},R_2^\ZZ) =\nu_1(\ccyl{v^1})}$. By $\sigma$-additivity of $\gamma^n$ and $\mu$, this equality holds for any $v^1\in R^{2m+1}$ with $m\leq n+k$. Symmetrically we have ${\gamma^n(R_1^\ZZ,\ccyl{v^2}) =\mu_{U_2}(\ccyl{v^2})}$ for any $v^1\in R^{2m+1}$.
Moreover, by definition, $\gamma^n(\ccyl{v^1},\ccyl{v^2}) = 0$ if there is no $u$ such that ${v^i\in P_u^i}$ for $i=1,2$. We deduce that the set:
\[X_n = \bigcup_{u\in Q^{2n+1}} S_u^1\times S_u^2\]
has measure $1$. More precisely, since $X_{n+1}\subseteq X_n$, for any ${m\leq n}$, ${\gamma^n(X_m)=1}$.\\
To conclude the proof, let $\gamma$ be any limit point of the sequence $(\gamma^n)_n$. By the definition of the distance on the space of measures, we have:
\begin{enumerate}
\item $\forall m,\forall w\in R_1^{2m+1}, \gamma(\ccyl{w},R_2^\ZZ) = \nu_1(\ccyl{w})$ and symmetrically for the $R_2$ component, hence $\gamma$ is a coupling of uniform measure on $R_1^\ZZ$ and $R_2^\ZZ$;
\item $\forall n,\gamma(X_n)=1$ hence $\gamma(\cap_n X_n) = 1$ where
\[\bigcap_n X_n = X = \{(s_1,s_2) \in R_1^\ZZ\times R_2^\ZZ : F_1(c,s_1)=F_2(c,s_2)\}\]
\end{enumerate}
We deduce that $\CAA_1$ and $\CAA_2$ are coupled on $c$ by measure $\gamma$.
\end{proof}
}

Notice that the proof of this theorem is non-constructive (recall that equality of stochastic global maps is undecidable in dimension~$2$ and higher). Moreover, it is easy to get convinced on a simple example that the coupling must depend on the configuration. Consider the two following automata with states $Q=R=\{0,1\}$ and neighborhoods $V=V'=\{0\}$: \CAA~with explicit global function $F(c,s)=s$ and \CAB\/~with explicit global function $G(c',s')=c'+s'\mod 2$. Clearly, both \CAA\/ and \CAB\/ define the same blank noise CA and the coupling proving this fact is defined for all $z\in\ZZ$ and all $a,b\in\{0,1\}$ by $\gamma_c(\cyl az,\cyl bz) = 1/2$ if and only if $a = b+c_z \mod 2$, and $=0$ otherwise. This coupling demonstrates indeed that $\gamma_c\bigl(\{(s,s'):F(c,s) = G(c,s')\}\bigr)=1$ yielding that the dynamics are identical; but note that $\gamma_c$ must depend on~$c$.

\DELETE{\begin{proof}(Sketch)
Given a configuration $c$, the principle is for every $n$ and every word $u$ of length $2n+1$ to build an sequence of increasing finite support couplings $(\gamma_c^n)$ where $\gamma_c^n$ matches the two finite subsets of finite random strings (thanks to the locality of \CAA\/ and \CAB)  that yields the word $u$ centered on $0$ in $\CAA$ and $\CAB$. We then conclude by extracting  from $(\gamma_c^n)$ a subsequence converging to some measure $\gamma_c$ thanks to the compacity of the metric space $\Mes{R_A^\ZZ\times R_B^\ZZ}$. As pointed out earlier, this proof is non-constructive.
\end{proof}}

\subsection{Other Undecidable Properties}
\label{sec:otherundeci}
Of course, stochastic CA inherit many undecidable properties from deterministic CA since they are a generalization of them. However, with the stochastic formalism new global properties can be considered together with their associated decision problem. For instance, we say that a CA $F$ is \textit{noisy} if it may reach any configuration from any configuration, i.e. if ${\NDET{F}(c)=Q^\Z}$ for all $c\in Q^\Z$. 
Surprisingly, this basic property is undecidable in dimension two.
 
\begin{proposition}
 It is undecidable to determine whether a given CA of dimension $d\geq2$ is noisy.
\end{proposition}
\begin{proof}
 Using the construction of Theorem~\ref{thm:deciding}, we show that the surjectivity problem reduces to testing noisiness. Hence it is undecidable starting from dimension two \cite{Kari94}.
\end{proof}
 
By contrast, testing whether a CA is deterministic or not is decidable in any dimension.
 
\begin{proposition}\label{prop:detdec}
 A CA $F$ of local function ${f:Q^\rho\times R^{\rho'}\rightarrow Q}$ is deterministic if and only if ${f(u,v)=f(u,v')}$ for all $u$,$v$ and $v'$. Hence testing whether a CA is deterministic is decidable in any dimension.
\end{proposition}
\begin{proof}
 If the condition on $f$ is verified then clearly $F$ is deterministic. Conversely, if there are some $v$ and $v'$ such that ${f(u,v)\not=f(u,v')}$ then it is straightforward to construct infinite configurations $c$ in $Q^\Z$ and $s_1$ and $s_2$ in $R^\Z$ such that ${F(c,s_1)\not=F(c,s_2)}$ contradicting determinism.
\end{proof}
 
The two undecidable problems presented so far (equality of global maps and noisiness) were shown undecidable for general stochastic CA starting from dimension 2. Section \ref{sec:dimensionone} below shows how these problems become decidable when we restrict to one-dimensional CA. We now give an example of a basic problem which is undecidable starting from dimension 1.
 
\paragraph{Pattern Probability Threshold Problem.} Given a SCA $F$, a language $L$ over $Q_F$ and a threshold function $\vartheta:L\rightarrow [0,1]$, the problem is to determine whether there is an initial configuration $c$ and a word $u\in L$ such that the image of $c$ by $F$ matches \emph{in one step} the word $u$ at position $0$ with a probability above $\vartheta(u)$. We will typically ask $\vartheta$ to decrease exponentially with the length of $u$. 
The problem is parametrized by $L$ and $\vartheta$ and is too general without additional restriction. In the sequel we will focus on the following restriction of this problem where the patterns to be matched are of the form $x\cdot y^+\cdot z$ and $\vartheta$ is exponentially decreasing in the length of the pattern, called \PPT{}:
 
\begin{definition}
The problem \PPT{} (Pattern Probability Threshold) is the following:
 \begin{description}
 \item[\quad Input: ] A stochastic CA $F$ over a state set $Q$ with $|Q|\geq4$, three distinguished distinct symbols $x,y,z\in Q$, and a threshold function $\vartheta: \mathbb N \rightarrow [0,1]$
 \item[\quad Question: ] Is there a configuration $c$ and
   $n\geq 1$ such that
   ${\bigl(\STOC{F}(c)\bigr)(\cyl{x\cdot y^n\cdot z}{0})> \vartheta(n)}$?
 \end{description}
\end{definition}
 
A threshold function $\vartheta$ is \emph{computationally superexponential} if:
\begin{enumerate}
\item $\vartheta(n) = \omega(\lambda^n)$ for all $\lambda <1$;
\item there is an algorithm that given any $0<\mu<1$ outputs a $K$ such that for all $n> K$,
  $\vartheta(n) > \mu^n$
\end{enumerate}
\begin{theorem}~%
 \label{thm:pptundeci}
 \begin{enumerate}
 \item The problem \PPT{} is undecidable for a stochastic CA $F$ when
   $\vartheta(n) = \frac12 \frac{2^{n}}{|R_F|^n}$, \emph{even when restricted to
     dimension 1}.
 \item If the threshold $\vartheta$ is non-increasing and
   computationally superexponential and with a computable limit, then
   \PPT{} is decidable for stochastic CA in dimension~$1$.
 \end{enumerate}
\end{theorem}
 
The proof of this theorem relies on the undecidability of the existence of a word recognized with a probability higher than some fixed threshold (namely, $\frac12$ as we will see later on) by a given probabilistic finite automata $\mathcal A$ \cite{GimbertO10}. The key is to encode into a stochastic CA the complete recognition process of the word written in the initial configuration $c$ by $\mathcal A$ in one single step of the SCA, so that the word $u=c_1\ldots c_n$ is recognized by $\mathcal A$ if and only if the pattern $x\cdot y^{n}\cdot z$ appears in the image configuration at position $0$. We will ensure that this happens with probability exactly $\Pr\{\text{$u$ is recognized by $\mathcal A$}\}/|Q|^n$ and thus the pattern $x\cdot y^n \cdot z$ appears in the image configuration of $c$ with probability at least $1/2|Q|^n$ if and only if $u$ is recognized by $\mathcal A$ with probability at least $\frac12$, proving the undecidability of \PPT{}.

\newcommand\fsta{\mathcal F}
\paragraph{Encoding probabilistic finite automata into CA.} A \emph{probabilistic finite automaton} $\mathcal{A}$ consists in a quintuple ${(Q,A,(M_a)_{a\in A},I,\fsta)}$ where:
\begin{itemize}
\item $A$ is the finite alphabet and $Q$ the finite set of states;
\item $I\in Q$ is the initial state and $\fsta\subset Q$ the set of final states;
\item for each $a\in A$, ${M_a\in [0,1]^{Q\times Q}}$ is a stochastic matrix giving the transition probabilities:
 \begin{itemize}
 \item for each $a$ and each $q_1$, ${\displaystyle\sum_{q_2}M_a(q_1,q_2) = 1}$.
 \item $M_a(q_1,q_2)$ is the probability to go from state $q_1$ to state $q_2$ when reading letter $a$;
 \end{itemize}
\end{itemize}
 
An \emph{accepting path} in the automaton for a word $u\in A^n$ is a finite sequence ${(q_i)_{0\leq i\leq n}}$ of states such that:
\begin{itemize}
\item the first state $q_0=I$ is the initial state of the automaton,
\item $q_n\in \fsta$.
\end{itemize}
The \textit{weight} of a path for a word $u$, is the product of weights of the transitions of the path labeled by the succesive letters of $u$:
\[\prod_{1\leq i\leq n}M_{u_i}(q_{i-1},q_i).\]
The \textit{acceptance probability} of a word $u\in Q^n$ for $\mathcal{A}$, $\probfa{u}$, is the sum of the weights of all accepting paths for $u$:
\[\probfa{u} = \sum_{
 \begin{matrix}
   q\in Q^{n+1}\\
   q_0= I,\, q_n\in\fsta
 \end{matrix}
}\prod_{1\leq i\leq n}M_{u_i}(q_{i-1},q_i). \]

\begin{proposition}
 \label{prop:encodePFA}
 For any probabilistic finite automaton $\mathcal A$ with rational transition probabilities, state set $Q$ and alphabet $A$, there exists a stochastic CA $F$ of dimension $1$ with states ${Q_F = A\sqcup\{\Qi,\Qo,\Qf,\Qe\}}$ and random states $R_F=Q\times \{1,\ldots,m\}$ (where $m$ is the least comon multiple of the denominators of all transition probabilities) such that for all ${n\geq 1}$:
 \begin{enumerate}
 \item ${\bigl(\STOC{F}(c)\bigr)(\cyl{\Qi\cdot \Qo^n\cdot \Qf}{0}) = 0}$ if ${c\not\in\cyl{\Qi\cdot A^n\cdot \Qf}{0}}$,
 \item ${\bigl(\STOC{F}(c)\bigr)(\cyl{\Qi\cdot \Qo^n\cdot \Qf}{0}) = \probfa{u}/|Q|^n}$ if ${c\in\cyl{\Qi\cdot u\cdot \Qf}{0}}$ with $u\in A^n$
 \end{enumerate}
Moreover, $F$ can be obtained algorithmically from $\mathcal A$.
\end{proposition}
 
\begin{proof}
Informally, the random symbols will try to guess an accepting path for the word written in the initial configuration. Each random symbol will be our guess of the transition taking place in the automaton while reading the letter at that position in the initial configuration. Checking if this guess is correct can be done locally. It the path is indeed correct and accepting, then the SCA prints the configuration $\Qi \Qo^n\Qf$; otherwise, it prints a $\Qe$ at each place an error occurs. The exponential decay comes from the fact that the SCA needs to guess what is the state just before reading each letter, which implies a multiplicative shift of the success probability by $1/|Q|^n$.
 
Formally, let $m$ be the least common multiple of the denominators of all transition probabilities and $R=Q\times\{1,\ldots,m\}$. Let $\tau:A\times Q\times\{1,\ldots,m\} \rightarrow Q$ an arbitrary function such that for all $(q,q')\in Q^2$, ${\#\{1\leq i\leq m: \tau(a,q,i) = q'\} = m \cdot M_a(q,q')}$ --- it exists since all $M_a(q,q')$ are integer multiples of $1/m$ and $\sum_{q'\in Q} M_a(q,q')=1$ for all $(a,q)\in A\times Q$. Intuitively, when a random symbol $(q,i)$ is uniformly selected in $R$ at a cell with symbol $a$ in the configuration, the SCA reads it as the guess that the transition that probabilistic automata follows when reading this letter $a$ is $q\xrightarrow{a} \tau(a,q,i)=q'$, where $q$ is chosen uniformly at random and $q'$ is chosen with probability $M_a(q,q')$. Every given accepting path $q_0=I,\ldots,q_n\in\fsta$ for a given word $u\in A^n$ written on the initial configuration is thus correctly guessed by the SCA with probability $1/|Q|^n$ (for guessing the right sequence of states $q_1,\ldots,q_n$) multiplied by $\prod_{i=1}^{n} M_{u_i}(q_{i-1},q_i) = \probfa{u}$ for guessing the right transition at each letter. Let us now describe precisely the SCA.
 
\newcommand\good[1]{\operatorname{\mathsf{good}}(#1)}
\newcommand\inifin[1]{\operatorname{\mathsf{valid}}(#1)}
Next, define the local function ${f: Q_F^{\{-1,0\}}\times R^{\{-1,0\}}\rightarrow Q_F}$ of $F$ as follows:
\[f(c_{-1}c_0,\underbrace{(q_{-1},i_{-1})}_{=s_{-1}}\underbrace{(q_0,i_0)}_{=s_0}) = 
 \begin{cases}
 \Qi & \text{if $c_0 = \Qi$},\\
  \Qo & \text{if $c_0\in A$ and $q_0 = \tau(c_0,\gamma,i_{-1})$ with $\gamma = \begin{cases} I & \text{if $c_{-1} = \Qi$,}\\ q_{-1} & \text{otherwise}\end{cases}$}\\
   \Qf &\text{if $c_0=\Qf$ and $q_{-1}\in\fsta$},\\
   \Qe &\text{in every other case.}
 \end{cases}
\]
By construction, $F(c,s)\in\cyl{\Qi\cdot \Qo^n\cdot \Qf}{0}$ if and only if:
\begin{enumerate}
\item $c\in\cyl{\Qi\cdot A^n\cdot \Qf}{0}$,  and
\item $q_1 = \tau(c_1,I,i_0)$ and $q_k = \tau(c_k, q_{k-1}, i_{k-1})$ for $2\leq k \leq n$, and
\item $q_n\in\fsta$,
\end{enumerate}
where $s_k = (q_k,i_k)$ for $k\in\mathbb N$. Fix a configuration $c\in\cyl{\Qi\cdot A^n\cdot \Qf}{0}$, and let $u=c_1\ldots c_n$. Then,
\begin{align*}
\bigl(\STOC{F}(c)\bigr)&(\cyl{\Qi\cdot \Qo^n\cdot \Qf}{0}) 
	 = \frac{
		\#\left\{ ((q_0,i_0),\ldots,(q_n,i_n)) \in R^{n+1}  \left| \begin{array}{l}  q_1 = \tau(u_1,I,i_0) \text{ and }  q_n\in\fsta \text{, and}\\[1mm]  q_k = \tau(u_k, q_{k-1}, i_{k-1}) \text{ for } 2\leq k \leq n\end{array}\right.\right\}}{|R|^{n+1}}\\
	& = \frac{1}{|R|^n} 
		\sum_{(q_1,\ldots,q_n)\in Q^{n-1}\times \fsta} 
			\#\{ i_0 : \tau(u_1,I,i_0) = q_1\} \times 
			\prod_{k=1}^{n-1} \#\{ i_k : \tau(u_{k+1},q_k,i_k) = q_{k+1}\}\\
	& = \frac{1}{|R|^n} 
		\sum_{(q_0,q_1,\ldots,q_n)\in \{I\}\times Q^{n-1}\times \fsta}
			\prod_{k=1}^n m\cdot M_{u_k}(q_{k-1},q_k)\\
	& = \frac{1}{|Q|^n} \probfa{u},	
\end{align*}
by definition of $\tau$, which concludes the proof.
\end{proof}
 
Using classical undecidability results concerning acceptance threshold in probabilistic finite automata we can now prove the Theorem stated earlier.
 
\begin{prOOf}{of the undecidability result in Theorem \ref{thm:pptundeci}}
 In \cite{GimbertO10}, the following problem is shown undecidable:
 \begin{description}
 \item[input:] a probabilistic finite automaton with all transition probability in ${\{0,\frac{1}{2},1\}}$
 \item[question:] is there a word accepted with probability at least $\frac{1}{2}$?
 \end{description}
 Using Proposition~\ref{prop:encodePFA} it is straightforward to check that this problem reduces to problem \PPT{}.
\end{prOOf}
 
Let us now show that if the threshold is superexponential in $n$, then one can decide \PPT{} in dimension~$1$. Assume that the threshold function $\vartheta$ is non-increasing and superexponential (i.e. $\vartheta(n) = \omega(\lambda^n)$ for all $\lambda <1$). We will show the following structural lemma. Let us say that a word $w\in Q^{n+\ell-1}$ is \emph{$\ell$-looping} if $w = a l a$ for some $a\in Q^k$ and $l\in Q^n$ (recall that $\ell = 2k+1$ where $k$ is radius of the considered SCA). We say that a word $w\in Q^\ell$ is \emph{non-deterministic} for a SCA $F$ if there are two random words $s,s'\in R^\ell$ such that $f(w,s) \neq f(w,s')$, and \emph{deterministic} otherwise. By extension, we say a word $w\in Q^{n+\ell-1}$ is deterministic for $F$ if all the subwords of length $\ell$ it contains are deterministic for $F$.

\newcommand\probdonne[2]{\Pr\{#1\overset{F}{\rightarrow} #2\}}
To simplify notations, we denote by $\probdonne{v}{w}$ the probability ${\STOC F(\cyl{v}{-k})(\cyl{w}{0})}$ when $v$ and $w$ have appropriate lengths (${|v|=|w|+2k}$).
 
\begin{lemma} \label{lem:ppt:superexp}
If there are $n\in\NN$ and a word $u\in Q^{n+\ell+1}$ such that ${\probdonne{u}{xy^nz} > \vartheta(n)}$, then there are $n'\in \NN$ and a word $u'\in Q^{n'+\ell+1}$ such that $\probdonne{u'}{xy^{n'}\!z}>\vartheta(n')$ with either
\begin{itemize}
\item $n'\leq K$, or
\item $u' = \gamma a(la)^q \gamma'$ where $|\gamma|+|\gamma'| \leq K$, $a\in Q^k$, $|l|\leq |Q|^\ell$, $q\in\NN^*$,  and the word $ala$ is deterministic for $F$. 
\end{itemize}
where $K$ is a constant that can be algorithmically computed from $\ell$, $|Q|$, $|R|$ and $\vartheta$.
\end{lemma}
 
\begin{proof}

First note that by the pigeonhole principle, any word in $Q^*$ of length at least $\ell+|Q|^\ell$ contains a $\ell$-looping subword $ala$ with $a\in Q^k$ and $|l|\leq|Q|^\ell$. Let us isolate the centerpart of $u$ by writing $u=aba'$ where $a,a'\in Q^k$. Let us mark in $b$ all the letters which are at the center of non-deterministic neighbourhoods. The marked letters split $b$ into subwords where each letter belongs to a deterministic neighbourhood. 
 
	Let us first assume that all of this deterministic subwords have length  at most $\ell+|Q|^\ell$, then there are at least $n/(\ell+|Q|^\ell)$ distinct letters at the center of non-deterministic neighbourhoods in $b$. Since neighbourhoods at distance at least $\ell$ from each other evolve independently, and since every non-deterministic neighbourhood produces an error in the pattern with probability at least $1/|R|^\ell$,  it follows that the image of $u$ by $F$ will be $xy^nz$ with probability at most $(1-1/|R|^\ell)^{n/\ell(\ell+|Q|^\ell)} = \mu^n$ where $\mu = (1-1/|R|^\ell)^{1/\ell(\ell+|Q|^\ell)}$. As $\probdonne{u}{xy^nz} > \vartheta(n)$, it follows that $\mu^n < \vartheta(n)$ but since $\vartheta$ is superexponential and non-increasing, there is a constant $K$ that can be algorithmically computed from $|Q|, |R|, \ell$ and $\vartheta$ (assuming an appropriate oracle for the superexponentiallity of $\vartheta$)
such that for all $m>K$, $\vartheta(m) > \mu^m$. It follows that $n \leq K$ and case 1 is verified for~${u'=u}$.
	
	Let us now assume that one of the deterministic (unmarked) subwords of $b$ has length at least $\ell+|Q|^\ell$, it follows that it contains a $\ell$-looping subword of length at most $|Q|^\ell$. Note that one can strip or duplicate the loop $la$ in any $\ell$-looping subword $ala$ in $u$: this will only affect the image (by adding or deleting some $y$s in the image) but will not change the probability of obtaining the pattern $xy^*z$. We then strip in $u$ all the loops in the $\ell$-looping deterministic subwords but one and duplicate the remaining one as many times as necessary to obtain a word $u'$ at least as long as $u$, i.e. so that $u'$ has length $\ell+n'+1$ with $n'\geq n$. Since $\vartheta$ is non-increasing and since the probability to obtain $xy^{n'}\!z$ from $c'$ is identical to the probability to obtain $xy^nz$ from $c$, we have $\probdonne{u'}{xy^{n'}\!z} >\vartheta(n)\geq\vartheta(n')$. Now, $u'$ has the form $\gamma a(la)^q\gamma'$ where $ala$ is deterministic and all the deterministic subwords in $\gamma$ and $\gamma'$ have length at most $\ell+|Q|^\ell$. Using the same argument as before, the sum of the lengths of the two words $\gamma$ and $\gamma'$ is bounded by the constant $K$ and therefore $u'$ has the desired properties for case 2.
\end{proof}
 
\begin{prOOf}{of the decidability part of Theorem~\ref{thm:pptundeci}}
Let $\theta=\lim_{n\rightarrow\infty} \vartheta(n)$.
Consider a SCA $F$ with state set $Q$, random symbol set $R$, radius $k$ and neighborhood width $\ell=2k+1$. Let $\mu = (1-1/|R|^\ell)^{1/\ell(\ell+|Q|^\ell)}$ and compute $K$ such that $\vartheta(m)>\mu^m$ for all $m\geq K$. According to Lemma~\ref{lem:ppt:superexp},  there are an $n\in\NN$ and a word $u\in Q^{n+\ell+1}$ such that $\probdonne{u}{xy^nz}>\vartheta(n)$ if and only if there is such a pair with $n\leq K$ or there is a deterministic $\ell$-looping word $ala$ of length at most $\ell+|Q|^\ell$ and two words $\gamma, \gamma'$ of total length at most $K$ and a $q\in\NN$ such that $u=\gamma a (la)^q\gamma'$ verifies the condition. One can easily check these conditions by enumerating all the subwords of length at most $K$ and loops of length $|Q|^\ell$, the only difficulty consists in computing the right value for $q$ in the second case. 
This is achieved as follows: first compute the word $u=\gamma ala\gamma'$ with $|\gamma|+|\gamma'|\leq K$ and $|la|\leq |Q|^\ell$ such that $ala$ is deterministic for $F$ and for which $\probdonne{u}{xy^{|u|-2k-2}z}$ is maximum. If $\probdonne{u}{xy^{|u|-2k-2}z}\leq \theta$, then we conclude that the second possibility is not possible because pumping in the loop $la$ does not change the probability to obtain the pattern $xy^\ast z$ since $(la)$ is deterministic. 
If $\probdonne{u}{xy^{|u|-2k-2}z} > \theta$, then using the monotonicity of $\vartheta$, there must exist some $q$ such that $\probdonne{\gamma a(la)^q\gamma'}{xy^{n_q}z} > \vartheta(n_q)$ where $n_q = |\gamma|+|a|+q\cdot|la|+|\gamma'|-1-\ell$ (again because pumping in the deterministic loop $la$ does not change probabilities). Then the second case is verified.
\end{prOOf}

\subsection{About the Garden of Eden Theorem}

One of the most celebrated theorems in the setting of deterministic CA states that surjectivity is equivalent to pre-injectivity (Garden-of-Eden Theorem, \cite{Coornaert10}).
In the case of stochastic CA the situation is different as we will show in this section.

\begin{definition}
  Let $F : Q^\Z\times R^\Z\rightarrow Q^\Z$ be a stochastic CA.
  \begin{itemize}
  \item $F$ is \emph{surjective} if for any $c\in Q^\Z$ there is some $c'\in Q^\Z$ and some $s\in R^\Z$ such that ${F(c',s)=c}$.
  \item $F$ is \emph{injective} if for any $c_1,c_2\in Q^\Z$ and $s_1,s_2\in R^\Z$ we have
    \[F(c_1,s_1) = F(c_2,s_2) \Rightarrow c_1=c_2\]
  \item $F$ is \emph{pre-injective} if for any $c_1,c_2\in Q^\Z$ with finitely many differences and $s_1,s_2\in R^\Z$ we have
    \[F(c_1,s_1) = F(c_2,s_2) \Rightarrow c_1=c_2\]
  \end{itemize}
\end{definition}

Some remarks:
\begin{enumerate}
\item the definitions above agree with the classical deterministic setting when the stochastic CA considered turns out to be deterministic.
\item in general a stochastic system can be injective yet non-deterministic, for instance consider the following stochastic map ${f:[0,1]\rightarrow [0,1]}$ s.t.
  \[f(x) =
  \begin{cases}
    x/3 &\text{ with probability $1/2$}\\
    1/2+x/3 &\text{ with probability $1/2$}
  \end{cases}\]
\item it is easy to define some stochastic CA which is surjective but not injective (and not pre-injective), for instance ${F:Q^\Z\times Q^\Z\rightarrow Q^\Z}$ s.t.
  \[F(c,s) = s\]
\end{enumerate}

Contrary to the deterministic setting, there is no Garden-of-Eden Theorem for stochastic CA. However, there are still strong relationships between the notion defined above.

\begin{theorem}\label{thm:injec}
  Let $F$ be a stochastic CA, then we have:
  \begin{itemize}
  \item if $F$ is pre-injective then $F$ is surjective;
  \item if $F$ is injective then $F$ is deterministic.
  \end{itemize}
\end{theorem}
\begin{proof}
\extendstoalldimensions{}  For any $n>0$ and any word $w\in R^n$ we define the \textbf{deterministic} CA $F_w : (Q^n)^\Z\rightarrow (Q^n)^\Z$ as the grouping by groups of size $n$ of the map:
  \[c\in Q^\Z\mapsto F(c,\overline{w})\]
  where $\overline{w}$ is the periodic configuration of period $w$ (with $w_0$ on cell $0$).

  First, since $F$ is pre-injective, we have that, for any $w$, $F_w$ is pre-injective (straightforward). Hence, by choosing some ${r\in R}$, we deduce that $F_{r}$ is surjective (by the Garden-of-Eden theorem). Therefore $F$ is also surjective, which proves the first assertion of the theorem.

  Now suppose in addition that $F$ is injective. Then for any $n>0$ and any $w_1,w_2\in R^n$ we have $F_{w_1} = F_{w_2}$. Indeed, if it was not the case we would have some $c\in (Q^n)^\Z$ such that ${F_{w_1}(c)\neq F_{w_2}(c)}$. But since $F_{w_2}$ is surjective (shown above) there would exist some $c'$ such that $F_{w_2}(c') = F_{w_1}(c)$. Since $c'$ must be different from $c$ (because ${F_{w_1}(c)\neq F_{w_2}(c)}$) this contradicts the injectivity of $F$ (grouping is a bijective operation and does not affect injectivity).

To conclude the proof it is sufficient to take $n$ large enough (larger than $2k+1$ where $k$ is the radius of $F$): in this case $F_{w_1} = F_{w_2}$ for any $w_1,w_2$ means that $F$ does not depend on its $R$-component, hence it is deterministic.
\end{proof}

\section{Correlation-free local rules are simpler}
\label{sec:correlationfree}

\newcommand\cfl{\mathbb{P}_f}

A stochastic CA ${\mathcal A=(Q,R,V,V',f)}$ is \textit{correltation-free} if its neighborhood associated to the random component is trivial: ${V'=\{0\}}$ (see Definition~\ref{def:syntax}). Letting $\rho=|V|$, its local function $f$ is then of the form ${f:Q^k\times R\rightarrow Q}$ and it can be seen has a map $\cfl$ from $Q^\rho$ to probability distributions over $Q$ (maps from $Q$ to $[0,1]$ summing to $1$) as follows:
\[\cfl(q_1,\ldots,q_\rho): \gamma \mapsto \frac{\#\bigl\{\alpha\in R : f(q_1,\ldots,q_\rho,\alpha)=\gamma\bigr\}}{|R|}\]
Note that most of the literature concerning stochastic CA is restricted to local Correlation-Free CA and use map $\cfl$ to define them \cite{Toom,Gacs,Fates,Mairesse,RST,FatesRST06}.

As an immediate consequence of this form of local function, one can compute probabilities involved in the global function as a product of 'local probabilities' as shown by the following lemma. To simplify notations, it is stated in dimension 1 but extends without difficulty to higher dimensions.

\begin{lemma}
  \label{lem:cf}
  If $F$ is a local Correlation-Free stochastic CA of local function $f$ and radius $k$, we have for all configuration $c$ and all finite words $u$:
  \[\bigl(\STOC{F}(c)\bigr)(\cyl{u}{0}) = \prod_{0\leq z<|u|-1}\bigl(\cfl(c_{z-k},\ldots,c_{z+k})\bigr)(u_{z+1})\]
\end{lemma}
\begin{proof}
  It is sufficient to check that the set $\anevent_{c,\cyl{u}{0}}$ (see section~\ref{sec:stochdyn}) is a union of cylinders which is in one-to-one correspondence with the set
  \[\prod_{0\leq z<|u|-1}\bigl\{s:f(c_{z-k},\ldots,c_{z+k},s)=u_{z+1}\bigr\}.\]
  The lemma follows by application of the uniform measure on both sets.
\end{proof}

\subsection{Impossible behaviors}

We now present behaviors than can be realized by general stochastic CA but not by CFCA.

\subsubsection{Number-conserving CA}

Number-conserving CA are regularly used to model interacting particles (see \cite{DurandFR03,FormentiG03,FormentiKT08} for the case of deterministic CA). A classic example of interacting particles model is the usual random walk, which is number-conserving because the number of walkers is conserved. Again, we restrict to dimension 1 to simplify notations but the extension to higher dimension is straightforward.

\begin{definition}
A SCA $F$ is number-conserving if $Q=\{0\ldots q\}$ for some $q\in\mathbb{N}$ and for any finite configuration $c\in Q^{\mathbb{Z}}$ (i.e. a configuration with finitely many cells in a state other than $0$), we have 
\[(S_F(c))(\{c'\,:\,\sum_i c'_i=\sum_i c_i\})=1\]
where the infinite sums are well-defined because only finite configurations are considered.
\end{definition}

Note in particular that the definition implies ${\STOC{F}({}^\omega0^\omega)\bigl(\{{}^\omega0^\omega\}\big) =1}$ and more generally, when $c$ is a finite configuration and $c'$ is reachable from $c$ then ${\STOC{F}(c)(\{c'\})>0}$ (because there are only finitely many configurations reachable from $c$).

Remark that our definition requires  the number to be conserved almost surely and is thus more restrictive than the definition in \cite{Fuks2004}, which requires only the number to be conserved in expectation. The conclusion of \cite{Fuks2004}, leaves open the question of strictly conservative particle system in \CFCA{}. We settle this question by showing that there is \emph{no} \CFCA{} (\emph{nor powers} of \CFCA{}) that can simulate (surjective) conservative particle system.

First, remark that it is easy to design a SCA that simulates a conservative particle system. For instance, consider the following SCA $F$ with states $\{0,1\}$, random symbols $\{\leftarrow,\cdot,\rightarrow\}$, and radius~${k=2}$.  The $1$s represent the particles and the $0$s the empty cells. The random symbol represents the movement each particle is trying to make: stay for $\cdot$; move right for $\rightarrow$ to be performed if the right cell is $0$ and if this move does not induce any conflict with another particle;  move left for $\leftarrow$ to be performed if the left cell is $0$ and if this move does not induce any conflict with another particle.  Here is its local rule (we only give the neighbourhoods whose image are $1$, the others have image $0$):

\medskip

\newcommand{\RS}[1]{\makebox[1em][c]{$#1$}}
\newcommand{\RR}{{\RS{\scriptsize\rightarrow}}}
\newcommand{\RNR}{{\RS{\scriptsize\not\rightarrow}}}
\newcommand{\RL}{{\RS{\scriptsize\leftarrow}}}
\newcommand{\RNL}{{\RS{\scriptsize\,\not\!\leftarrow}}}
\newcommand{\RE}{{\RS*}}
\newcommand{\RD}{{\RS\cdot}}
\newcommand{\RA}{\RS{*}}
\newcommand{\RO}{\RS{0}}
\newcommand{\RI}{\RS{1}}
\newcommand{\FF}[2]{f\!\left(\!\!\begin{array}{cc}c:#1\\s:#2\end{array}\!\!\!\!\right)}

\centerline{\begin{minipage}{.8\textwidth}
\begin{multicols}{3}
\noindent $\FF{\RA\RA\RI\RA\RA}{\RE\RE\RD\RE\RE} = 1$\\[1mm] 
$\FF{\RA\RA\RI\RO\RI}{\RE\RE\RR\RE\RL} = 1$\\[1mm]
$\FF{\RI\RO\RI\RO\RA}{\RR\RE\RL\RE\RE} = 1$\\[1mm]
$\FF{\RA\RI\RI\RA\RA}{\RE\RE\RL\RE\RE} = 1$\\[1mm]
$\FF{\RA\RA\RI\RI\RA}{\RE\RE\RR\RE\RE} = 1$\\[1mm]
$\FF{\RA\RI\RO\RO\RA}{\RE\RR\RE\RE\RE} = 1$\\[1mm]
$\FF{\RA\RO\RO\RI\RA}{\RE\RE\RE\RL\RE} = 1$\\[1mm]
$\FF{\RA\RI\RO\RI\RA}{\RE\RD\RE\RL\RE} = 1$\\[1mm]
$\FF{\RA\RI\RO\RI\RA}{\RE\RR\RE\RD\RE} = 1$
\end{multicols}
The images of all other patterns are $0$ ($*$ stands for an arbitrary symbol).
\end{minipage}}
\medskip

This SCA is clearly non-deterministic, number-conserving and surjective (each cell remains unchanged when its random symbol is $\cdot$).

Note that most interacting particle systems are not only number-conserving but also surjective. It turns out that no CFCA nor iterates of CFCA can express such systems.

\begin{lemma} \label{lem:CFCA:num:det}
If a CFCA is number conserving, then it is a deterministic map.
\end{lemma}

\begin{proof}
Assume by contradiction that $F$ is a number conserving CFCA which is not deterministic. Let $c$ be a finite configuration such that there are $c'\neq c''$ such that $\STOC{F}(c)(\{c'\})>0$ and $\STOC{F}(c)(\{c''\})>0$. Let $i$ be such that $c'_i\neq c''_i$. Then, since the updates are independent in CFCA: with positive probability, $c$ is mapped to $c'$ and with positive probability, $c$ is mapped to $\hat c'$ where $\hat c'_i = c''_i$ and $\hat c'_{j\neq i} = c'_j$. This is a contradiction since the total weight of $c'$ and $\hat c'$ differ by $c'_i-c''_i\neq 0$.
\end{proof}

\begin{theorem}\label{thm:CFCA:num:det}

Let $F$ be a CFCA. If $\STOC{F}^t$ is number-conserving and surjective for some $t>0$, then $F$ is deterministic.

\end{theorem}

\begin{proof}
Assume by contradiction that $F$ is non-deterministic. By Lemma~\ref{lem:CFCA:num:det}, $F$ is not number conserving. Therefore there are finite $c,c',c''$  such that $\STOC{F}(c)(\{c'\})>0$, $\STOC{F}(c)(\{c''\})>0$, and the weight of $c'$ and $c''$ differ. As $F^t$ is surjective, so is $F^{t-1}$ and let $d$ be a finite configuration such that $S^{t-1}_F(d)(\{c\})>0$. It follows that $S^t_F(d)(\{c'\})>0$ and $S^t_F(d)(\{c''\})>0$ which contradicts the fact that $F^t$ is number conserving.     
\end{proof}

Whereas the surjectivity constraint was not needed in the lemma, it is required for the theorem to hold. Indeed, the square of the CFCA illustrated in Fig.~\ref{fig:sq:number} is non-deterministic and number-conserving. This CFCA~$F$ has neighbourhood $\{-5,\ldots,6\}$ (a neighbourhood large enough to prevent unstable patterns from propagating). The only patterns yielding to a state change are:
$$
\begin{array}{ll}
f(00000\mathbf11000{*}{*}) = 0
&	f({*}{*}000\mathbf0001000) = 1\\[1mm] 
f({*}0000\mathbf011000{*}) = 0 \text{ or } 1  \text{ with probability } \frac12 
&	f({*}{*}000\mathbf1001000) = 0 
\end{array}
$$
A cell matching any other pattern remains unchanged ($*$ stands for an arbitrary symbol). Note that this CFCA is not surjective since no image configuration contains the pattern $0000011000$. 

\begin{figure}
\centerline{
\includegraphics[scale=0.4]{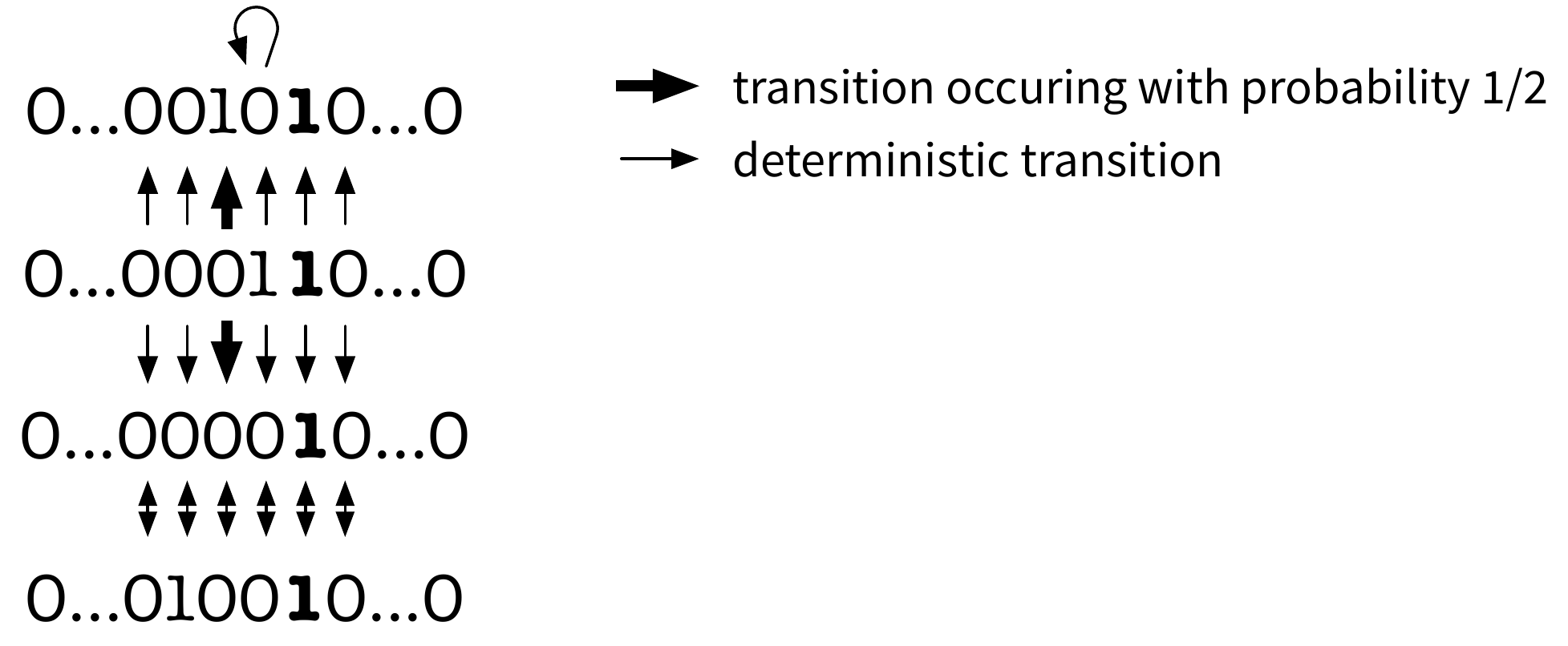}
}
\caption{An example of non-surjective non-deterministic CFCA whose square is number conserving.} \label{fig:sq:number}
\end{figure}

\subsubsection{Generation of random words of fixed parity}

Let us consider again the SCA \textsc{Parity} which was used as an example in Section~\ref{sec:stochdyn}.
It turns out that even iterates of \CFCA{} having $Q=\{\#,0,1\}$ cannot reproduce this behavior.

\begin{theorem}\label{th:reven}
For all $F$ a \CFCA{} with $Q_F=\{\#,0,1\}$ and for all $t$, we have $(\STOC{F})^t\neq \STOC{Parity}$.
\end{theorem}
\begin{proof}
  Assume by contradiction that there is a \CFCA{} $F$ and
  $t$ such that $(\STOC{F})^t= \STOC{Parity}$.  Let $k$ be the radius
  of $F$.  Both images of $0^{2k+1}$ and $1^{2k+1}$ cannot be
  deterministic. Otherwise, $F^t([0^{2kt+1}]_{-kt})_0$ would be a
  deterministic value which contradicts the random generation of words of even parity starting from the pattern ${\#0^{2kt+1}\#}$. Therefore $F(0^{2k+1})$ is $0$ with
  probability $p>0$, and $1$ with probability $1-p>0$.  Then, with
  positive probability $q$, $F^{t-1}([\#0^{2kt+1}\#]_{-rt})\subseteq
  [0^{2k+1}]_{-k}$ and then the central cell of $0^{2k+1}$ is mapped
  independently to $0$ or $1$ with probability $p$ and $1-p$. It
  follows that there is a probability at least $\min(qp,q(1-p))>0$
  that the word $\#0^{2kt+1}\#$ is mapped by $F^t$ to an odd parity
  word.
\end{proof}
\subsection{Decidability}

\begin{theorem}
Let $F$ be a \CFCA{}. It is decidable to test whether it is noisy.
\end{theorem}
\begin{proof}
  It is enough to check the local rule. Indeed, as opposed to SCA, a
  \CFCA{} is noisy if and only if every neighbourhood has a
  positive probability to be mapped to every letter.
\end{proof}

\CFCA{} do not have the same expressiveness as SCA. However, their squares can introduce correlations, and can in fact simulate any SCA (see Theorem \ref{thm:CFCAsimu}). In particular, deciding the equality of the stochastic functions of squares of CFCA is undecidable, as proven bellow:
\begin{theorem}
\label{thm:undecnoisecf}
Fix any dimension ${d\geq 2}$. It is undecidable whether a given CFCA $F$ of dimension $d$ is such that $F^2$ is a
white noise CA (i.e. ${\STOC{F}^2(c)}$ is the uniform measure
for all $c$). It is also undecidable whether $F^2$ is noisy.
\end{theorem}
\begin{proof}
  As in Theorem \ref{thm:deciding}, we reduce from undecidability of
  determining the surjectivity of CA in dimension $2$ and higher. Let
  $F$ be an automaton over ${\mathbb{Z}^d}$ with states $Q$. We define
  the \CFCA{} $G$ over ${\mathbb{Z}^d}$ with states
  $Q\times Q$ and random states $Q$ as follows:
  $G((c,c'),s)=(F(c'),s)$. Then,
  $G^2((c,c'),s,s')=G((F(c'),s),s')=(F(s),s')$. It follows that, as in
  Theorem \ref{thm:deciding}, $\STOC{G}^2$ is the uniform measure over
  $(Q\times Q)^{\mathbb{Z}^d}$, i.e. it is white noise, if and only if
  $F$ is surjective.

  The same encoding shows undecidability of deciding if the square of CFCA is noisy
  because $G^2$ is either white noise or non-noisy, whether $F$ is surjective or not.
\end{proof}

\begin{corollary}\label{cor:undeciteratedequalityforCFCA}
  It is undecidable to determine, given $F$ and $G$ two CFCA, whether ${\STOC{F}^2=\STOC{G}^2}$.
\end{corollary}
\begin{proof}
  From Theorem~\ref{thm:undecnoisecf} is is sufficient to consider particular instances of the problem where $G$ a fixed CFCA which is noisy.
\end{proof}

\section{Dimension 1 is simpler}
\label{sec:dimensionone}

In this section, we restrict to dimension $1$.

\subsection{Weighted De Bruijn Automata}
\newcommand\tauto[1]{\mathcal{A}_{#1}}

To any SCA $F$, syntactically given by $(Q,R,V,V,f)$, we associate a weighted finite automaton $\tauto{F}=(\Sigma,A,\delta,w,i_0,\fsta)$ whose weight are in $\QQ$ (see \cite{handbookWFA} for an introduction to weighted automata). Intuitively, $\mathcal{A}_F$ is a De Bruijn automaton recognizing pairs of configurations of the form $(c,\sigma^{k-1}\circ F(c,\cdot))$ and the weights correspond to probability distributions given by the stochastic global function $\STOC{F}$. More precisely, the shift of ${k-1}$ cells between the two components of the recognized pairs comes from the internal memory of the automaton which needs to be initialized as detailed below.

First, we suppose without loss of generality that the neighbourhoods are of the form $V=V'=\{-k,\ldots,k\}$. We let $\ell=2k+1$. $\tauto{F}$ works on alphabet $\Sigma=Q\cup Q\times Q$ and its set of states is
\[A = \bigcup_{0\leq j\leq \ell-1}Q^{j}\times R^{j}\]
By convention, we denote $i_0=(\epsilon,\epsilon)$ the only element of $Q^0\times R^0$.
The transition relation $\delta\subseteq A\times\Sigma\times A$ is given by:
\begin{itemize}
\item\textbf{(initialization)} for any $j<\ell-1$, any state $(\overline{a},\overline{b})\in Q^j\times R^j$, any $q\in Q$, and any $r\in R$ we have \[\bigl((\overline{a},\overline{b}),q,(\overline{a}q,\overline{b}r)\bigr)\in\delta\]
\item\textbf{(main component)} for any $\ell$-uples $q_1,\ldots, q_\ell\in Q$ and ${r_1,\ldots,r_\ell\in R}$, we let ${\alpha=\bigl(q_\ell,f(\overline{q},\overline{r})\bigr)}$ and we have
\[\bigl((q_1\cdots q_{\ell-1},r_1\cdots r_{\ell-1}), \alpha, (q_2\cdots q_{\ell},r_2\cdots r_{\ell})\bigr)\in\delta\]
\end{itemize}
Finally, we let $\fsta=A\setminus\{i_0\}$ and weights are $\frac{1}{|R|}$ for all transitions. 

A \emph{path} in $\tauto{F}$ is a sequence of transitions starting from $i_0$ and ending in $\fsta$, such that any transition starts from the state where the previous one arrives. The word recognized by a path is the sequence of labels of transitions. The weight of a path is the product of the weights of transitions. The weight of a word $u\in\Sigma^\ast$, denoted $w_{\tauto{F}}(u)$, is $0$ if $u$ is not recognized by any path, and the sum of weights of paths which recognize $u$ otherwise.

By construction, the automaton $\tauto{F}$ recognizes only words of the form ${Q^{\ell-1}\bigl(Q\times Q\bigr)^\ast}$. All transitions have the same weight. But, two recognized words of same length do not have the same weight in general.

\begin{lemma}
  \label{lem:autoca}
  Let $u\in Q^p$ and $v\in Q^{p+\ell-1}$ for some $p\geq 1$. Define $m = v_1\cdots v_{\ell-1}\cdot (v_\ell,u_1)\cdots (v_{p+\ell-1},u_p)$. Then, for any configuration $c\in\cyl{v}{0}$, we have ${w_{\tauto{F}}(m) = \bigl(\STOC{F}(c)\bigr)(\cyl{u}{k})}$.
\end{lemma}
\begin{proof}
  It is straightforward to check from the defintion of $\tauto{F}$ that the set of paths recognizing $m$ is in one-to-one correspondance with the set of word $\nu\in R^{p+k-1}$ such that ${F(c,\cyl{\nu}{0})\subseteq\cyl{u}{k}}$. Moreover, this set is exactly $\anevent_{c,\cyl{u}{k}}$ (see section~\ref{sec:stochdyn}) and the measure of this set for the uniform probability measure is by definition ${\bigl(\STOC{F}(c)\bigr)(\cyl{u}{k})}$. Finally, since each path of length $l$ has weight $\frac{1}{|R|^l}$, the weight $w_{\tauto{F}}(m)$ of $m$ is equal to the measure of $\anevent_{c,\cyl{u}{k}}$ for the uniform probability measure. The lemma follows.
\end{proof}

\begin{corollary}\label{cor:dec1Dequality}
  Equality of stochastic global function of 1D CA is decidable.
\end{corollary}
\begin{proof}
 Consider two CA $F$ and $G$. We can suppose that they have the same centered neighbourhood $V=\{-k,\cdots,k\}$ (if not simply increase syntactically neighbourhood apropriately). From lemma~\ref{lem:autoca} it follows that equality of $\STOC{F}$ and $\STOC{G}$ is equivalent to the equality of $w_{\tauto{F}}$ and $w_{\tauto{G}}$, i.e. equality of the weighted languages of $\tauto{F}$ and $\tauto{G}$. The problem of equivalence of weighted finite automata is decidable for weights in $\QQ$ \cite{handbookWFA}. Since $\tauto{F}$ and $\tauto{G}$ are computable from $F$ and $G$, the corollary follows.
\end{proof}

\subsection{Simplified automaton for Correlation-Free CA}
\newcommand\tautobis[1]{\mathcal{B}_{#1}}

Consider a \CFCA{} $F$ , with ${\mathcal A=(Q,R,V,\{0\},f)}$. The construction detailed in the above subsection gives a weighted automaton ${\tauto{F}=(\Sigma,A,\delta,w,i_0,\fsta)}$ which has some additional regularities due to correlation-freeness. Intuitively, memorizing the $R$-component in states of the automaton is useless. We now construct a \textit{deterministic} weighted finite automaton ${\tautobis{F}=(\Sigma,B,\delta',w',i_0',\fsta')}$ which is equivalent to $\tauto{F}$. $\tautobis{F}$ is essentially a De Bruijn graph with an initialisation part. The weights are given by the map $\cfl$ associated to the local function $f$ of $F$. Formally, $\tautobis{F}$ is defined as follows (again with $\ell=2k+1$ and $\Sigma=Q\cup Q\times Q$):
\begin{itemize}
\item $\displaystyle B = \bigcup_{0\leq i<k} Q^i$;
\item $i_0'=\epsilon$ (the single element of $Q^0$) and $F = Q^{\ell-1}$;
\item ${w': B\times\Sigma\times B\rightarrow [0,1]}$ gives implicitly $\delta'$ (all non-zero weight transitions) and is defined by:
  \begin{itemize}
  \item for any ${i<\ell-1}$, any ${\overline{a}\in Q^i}$, any $q\in Q$ we have:
    \[w'\bigl(\overline{a},q,(\overline{a}q)\bigr) = 1;\]
  \item for any ${q_1,\ldots, q_\ell\in Q}$ we have:
    \[w'\bigl((q_1,\ldots,q_{\ell-1}),(q_\ell,q),(q_2,\ldots,q_\ell)\bigr) = \bigl(\cfl(q_1,\ldots,q_\ell)\bigr)(q)\]
  \item any transition not mentioned above has weight $0$.
  \end{itemize}
\end{itemize}

We define ${w_{\tautobis{F}}:\Sigma^\ast\rightarrow [0,1]}$ as just the product of weights of the transitions of the unique path labelled by $u$. Like $\tauto{F}$, $\tautobis{F}$ can be used to compute the probabilities involved in $\STOC{F}$.

\begin{lemma}
  \label{lem:autobisca}
  Let $u\in Q^p$ and $v\in Q^{p+\ell-1}$ for some $p\geq 1$. Define ${m = v_1\cdots v_{\ell-1}\cdot (v_\ell,u_1)\cdots (v_{p+\ell-1},u_p)}$. Then, for any configuration $c\in\cyl{v}{0}$, we have ${w_{\tautobis{F}}(m) = \bigl(\STOC{F}(c)\bigr)(\cyl{u}{k})}$.
\end{lemma}
\begin{proof}
Straightforward from the construction of $\tautobis{F}$ and lemma~\ref{lem:cf}: there is only one path in $\tautobis{F}$ recognizing $m$ and, after the initialization part ($v_1\cdots v_{\ell-1}$) of weight $1$, the weight of the path is given exactly by the product of $\cfl(\cdot)$ appearing in lemma~\ref{lem:cf}.
\end{proof}

Armed with this construction, we will now reevalute the \PPT{} problem in the special case of \CFCA{}. Unlike in the general case of Theorem~\ref{thm:pptundeci}, the problem will turn out to be decidable. First, let us introduce some vocabulary and a structural lemma.\\
We say that a word $m=v_1\cdots v_{\ell-1}\cdot (v_\ell,u_\ell)\cdots (v_{\ell+p-1},u_{\ell+p-1})$ of $\Sigma^*$ is \emph{valid} if $u_\ell=\Qi$, ${u_{\ell+1}=\cdots=u_{\ell+p-2}={\Qo}}$, $u_{\ell+p-1}=\Qf$, and  ${w_{\tautobis{F}}(m)>0}$. We say that a word is \emph{short} if its length is at most $\ell-1+|Q|^{\ell-1}+3$, i.e.  the number of states of the automaton plus three.
We say that a word $m=v_1\cdots v_{\ell-1}(v_\ell,u_1)\cdots(v_{\ell+q-1},u_q)\in\Sigma^*$ \emph{contains a loop} of length $\kappa$ at position~$i$, with $\ell\leq i \leq |m|-\kappa$, if $v_{i-\ell+j} = v_{i+\kappa-\ell+j}$ for all $1\leq j \leq \ell-1$; in this case, we say that the word $l=(v_{i},u_{i-\ell+1})\cdots(v_{i+\kappa-1},u_{i-\ell+\kappa})$ is a \emph{loop} of $m$ at position~$i$. The quantity $\bigl(w_{\tautobis{F}}(v_{i+\kappa-\ell+1}\cdots v_{i+\kappa-1}(v_{i},u_{i-\ell+1})\cdots(v_{i+\kappa-1},u_{i-\ell+\kappa}))\bigr)^{1/\kappa}$ is conveniently called the \emph{linear weight} $lw_{\tautobis{F}}(l)$  of the loop $l$. Lastly, we say that a loop $l = (v_1,u_1)\cdots(v_{\kappa},u_{\kappa})$ is a \emph{valid loop} if $u_1=\cdots=u_{\kappa} = {\Qo}$. Note that if a valid word contains a loop $l$ at position~$i$ with $\ell+1< i< |m|-|l|$, then for all $q\in \mathbb N$, the word $m_{1..i-1}l^q m_{i+|l|..|m|}$ is valid and has weight $w_{\tautobis{F}}(m)\times lw_{\tautobis{F}}(l)^{(q-1)|l|}$.
 
\begin{lemma} \label{lem:dec:CFCA:1D}
Assume there is a valid word $m'$, of length $\ell+p'-1$, having weight ${w_{\tautobis{F}}(m')> \alpha \lambda^{p'}}$ for some $\alpha>0$ and $\lambda <1$. Then there must also be a valid word ${m = v_1\cdots v_{\ell-1}\cdot (v_\ell,u_\ell)\cdots (v_{\ell+p+1},u_{\ell+p-1})}$, of length $\ell+p+1$, with weight ${w_{\tautobis{F}}(m)> \alpha \lambda^{p}}$ and such that at least one of the two following properties holds:
\begin{itemize}
\item
$m$ is short. 
\item
$m$ consists in three parts: $m = m_1 l^q m_2$ where
	\begin{itemize}
	\item 
	${m_1 = v_1\cdots v_{\ell-1} \cdot (v_\ell,u_\ell)\cdots(v_{\ell+a-1},u_{\ell+a-1})}$ with $2\leq a \leq |Q|^{\ell-1}+2$ 
	\item 
	$q\in\mathbb N$ and $l$ is a valid loop of length at most $|Q|^{\ell-1}$ and with linear weight $> \lambda$, and 
	\item
	$m_2 = {(v_{\ell+p-b},u_{\ell+p-b})\cdots(v_{\ell+p-1},u_{\ell+p-1})}$ with $1\leq b\leq|Q|^{\ell-1}$.
	\end{itemize}
\end{itemize}
\end{lemma}
\begin{prOOf}{of Lemma~\ref{lem:dec:CFCA:1D}}
Consider a valid word $m' = v_1\cdots v_{\ell-1}\cdot (v_\ell,u_\ell)\cdots (v_{p'+\ell-1},u_{p'+\ell-1})$  with weight ${w_{\tautobis{F}}(m')> \alpha \lambda^{p'}}$.
If $m'$ is short, take $m=m'$.\\
Otherwise, consider the subword $n'=(v_{2+\ell},u_{2+\ell})\cdots (v_{p'+\ell-2},u_{p'+\ell-2})$, of length $p'-3$. Since $m'$ is not short, $p'+\ell-1\geq \ell-1+|Q|^{\ell-1}+3$, hence $p'-3\geq|Q|^{\ell-1}$, therefore $n'$ must contain some loop $l'$ of length $\kappa'$. Since $m'$ is valid we have $u_\ell=\Qi$, $u_{1+\ell}=\Qo$, $u_{p'+\ell-1}=\Qf$, and the rest of the $u$'s are equal to $\Qo$, i.e. those of $n'$ are equal to $\Qo$ and so the loop is valid. We construct $m''$ by recursively removing all loops $l'$ that have linear weight $\leq \lambda$. The weight of ${m'' = v_1\cdots v_{\ell-1}\cdot (v_\ell,u_\ell)\cdots (v_{p''+\ell-1},u_{p''+\ell-1})}$ verifies ${w_{\tautobis{F}}(m'')> \lambda^{p''}}$ as well, and $m''$ is valid. If $m''$ is short, take $m=m''$.\\
Otherwise, consider the subword $n''=(v_{2+\ell},u_{2+\ell})\cdots (v_{p''+\ell-2},u_{p''+\ell-2})$ of $m''$, with length $p''-3$. Again $n''$ must contain some valid loop, and all its loops have linear weight $>\lambda$.
We construct $\tilde m$ from $m''$ by recursively removing all but one  loop, which we call $l$. This $\tilde m$ is again valid, and indeed of the form $m_1lm_2$ with $m_1$ and $m_2$ as specified in the lemma. Finally, since the linear weight of $l$ is $>\lambda$, one can choose $q\in\mathbb N$ so that the weight of the valid word $m=m_1l^qm_2$, of length $\ell-1+p$, verifies ${w_{\tautobis{F}}(m)> \alpha \lambda^p}$.
\end{prOOf}

The decidability result follows directly from this lemma.

\begin{theorem}\label{th:pptdeccfca}
  The problem \PPT{} is decidable for \CFCA{} of dimension 1 when the threshold function $\vartheta$ verifies $\vartheta(n) = \alpha \lambda^n$ for some $\alpha>0$ and $\lambda <1$ which may depends on the \CFCA.
\end{theorem}

\begin{proof}
We can decide \PPT{}
on input $F$ by the following algorithm:
  \begin{enumerate}
  \item Check if there is a short valid word $m$ verifying ${w_{\tautobis{F}}(m)> \alpha \lambda^{|m|-\ell-1}}$; if one is found answer \textbf{YES};
  \item If none is found, check if there is a valid word $m = m_1 l m_2$ with  $2\leq |m_1| \leq |Q|^{\ell-1}+2$, $1\leq |m_2| \leq|Q|^{\ell-1}$, and where $l$ is a valid loop of length at most $|Q|^{\ell-1}$ and with linear weight $>\lambda$; if one is found answer \textbf{YES} else answer \textbf{NO}.
  \end{enumerate}
\end{proof}

\subsection{Non-deterministic global function and model checking}

In this section we will briefly show how model checking methods based on B\"uchi automata and already developped for 1D deterministic CA \cite{Sutner09} can be extended to non-deterministic automata. It is not the purpose of our paper to give a self-contained exposition of this topic and we refer the reader to \cite{Sutner10,Finkel11} for more details and extensions of this approach.

The central concept here is that of \textit{$\omega$-automatic structure}. Intuitively it is a structure where objects are (semi-)infinite words over a finite alphabet $A$ and relations, seen as languages of (semi-)infinite words over $A\times A$, are all B\"uchi-recognizable. More generally, one can consider structures which are isomorphic to one of this form in order to allow other types of objects. In our case, objects are (bi-)infinite words and we consider the bijection ${\phi : A^\Z\rightarrow (A\times A)^\N}$ (for any finite set $A$) defined by:
\[\phi(c) = n\mapsto \bigl(c(n),c(-n)\bigr).\]

\begin{definition}
  A structure ${\bigl(A^\Z,(R_i)_{1\leq i\leq n}\bigr)}$, where each $R_i$ is a relation of finite arity $k_i$ over $A^\Z$, is \textit{$\omega$-automatic} if for each $i$ the following relation is B\"uchi-recognizable:
\[\bigl\{\bigl(\phi(c_1),\ldots,\phi(c_{k_i})\bigr) : (c_1,\ldots,c_{k_i})\in R_i\bigr\}.\]
Such a structure is finitely representable by the list of B\"uchi automata recognizing the relations.
\end{definition}

This definition is a particular case used in \cite{Finkel11} of a more general notion introduced in \cite{Hodgson} and further developped since then \cite{Kuske09}. The interest of $\omega$-automatic structure relies in their decidability. Strangely enough, this decidability result is not always stated in its uniform version which is the most useful in the context of cellular automata (while \cite{Finkel11} is not interested in uniformity, \cite{Sutner09} does not state uniform results but all the ingredients of uniformity are present in the proofs). 

\begin{theorem}
  \label{thm:deciomegaauto}
  The first-order theory of $\omega$-automatic structures is uniformly decidable, i.e. there is an algorithm that take has input an $\omega$-automatic structure, a first-order formula and decides whether the structure satisfies the formula.
\end{theorem}
\begin{proof}
  Although not stated exactly in this way, all details for the proof of this result are in \cite{Hodgson}, which is to our knowledge the first paper to introduce explicitely the notion of $\omega$-automatic structure. In fact, all the tools needed to prove this result (essentially closure properties of B\"uchi automata and decidability of the emptiness problem for recognized language) were already present in the seminal work of B\"uchi in the 60s.
\end{proof}

We can now state a number of decidability results. Some of them can be compared with undecidability results in dimension 2 and the last two are 'unquantified' variants of the \PPT{} problem which is also undecidable (see section~\ref{sec:otherundeci}).

\begin{corollary}\label{cor:buchi}
  The following problems are decidable for 1D stochastic CA:
  \begin{enumerate}
  \item given $F$, is it noisy?
  \item given $F$, is it surjective?
  \item given $F$, is it injective?
  \item given $F$, is it pre-injective?
  \item given $F$ and $G$, do we have ${\NDET{F}^t =\NDET{G}^t}$? (where $t$ is any fixed positive integer)
  \item given $F$, is there $u\in L$ and $c$ such that ${\NDET{F}(c)\subseteq\cyl{u}{0}}$? (where $L$ is any fixed rational language)
  \item given $F$, is there $u\in L$ and $c$ such that ${\NDET{F}(c)\cap\cyl{u}{0}\not=\emptyset}$? (where $L$ is any fixed rational language)
  \end{enumerate}
\end{corollary}
\begin{prOOf}{sketch}
  First, the following relations are B\"uchi-recognizable through encoding by $\phi$ (see \cite{Sutner09,Sutner10} for detailed proofs about all of them except the last one):
  \begin{itemize}
  \item ${x = y}$,
  \item ${x \approx y}$, i.e. ${x=y}$ up to a finite number of differences,
  \item ${x\rightarrow_F y}$, i.e. ${y\in\NDET{F}(x)}$,
  \item ${x\leadsto L}$, i.e. ${x\in\cyl{u}{0}}$ for some $u\in L$.
  \end{itemize}
  Then all problems above can be expressed as a first-order formula in the apropriate $\omega$-automatic structure:
  \begin{enumerate}
  \item ${\forall y,\forall x: x\rightarrow_F y}$
  \item ${\forall y,\exists x: x\rightarrow_F y}$
  \item ${\forall y, \forall x_1, \forall x_2: (x_1\rightarrow_F y \wedge x_2\rightarrow_F y)\Rightarrow x_1=x_2}$
  \item ${\forall y, \forall x_1, \forall x_2: (x_1\rightarrow_F y \wedge x_2\rightarrow_F y)\Rightarrow x_1\approx x_2}$
  \item ${\forall x, \forall y: \bigl(\exists z_1,\ldots, \exists z_{k-1}, x\rightarrow_F z_1\wedge\cdots\wedge z_{k-1}\rightarrow_F y\bigr)\Leftrightarrow \bigl(\exists z_1,\ldots, \exists z_{k-1}, x\rightarrow_G z_1\wedge\cdots\wedge z_{k-1}\rightarrow_G y\bigr)}$
  \item ${\exists x,\forall y : x\rightarrow_F y\Rightarrow y\leadsto L}$
  \item ${\exists x,\exists y : x\rightarrow_F y\wedge y\leadsto L}$
  \end{enumerate}
  Therefore, deciding one of the problem boils down to building the apropriate $\omega$-automatic structure and applying the algorithm of Theorem~\ref{thm:deciomegaauto}.
\end{prOOf}

\section{Intrinsic Simulation}
\label{sec:simul}

The purpose of this section is to give a precise meaning to the sentence ``$\AUTO{A}$ is able to simulate $\AUTO{B}$'' or equivalently ``$\AUTO{A}$ contains  the behavior of $\AUTO{B}$''. 

Our approach follows a series of works on simulations between classical deterministic CA \cite{phdrapaport,phdollinger,phdtheyssier,bulking1,bulking2}.
We are going to define simulation pre-orders on stochastic CA which extend the simulation pre-orders defined over classical deterministic CA in \cite{bulking2}. More precisely, we want the new pre-order to be exactly the classical pre-order when restricted to deterministic CA. For general background and motivation behind this simulation pre-order approach we refer to \cite{bulking1,bulking2}. Intrinsic simulation has also been brought to deterministic quantum CA in \cite{ArrighiNUQCA}.

In each case (the deterministic, the non-deterministic, and the stochastic global functions), we will define simulation as an equality of dynamics up to some \emph{local} transformations.

\textit{In this section and the next one, all CA considered are one-dimensional. There is no difficulty to extend all definitions and results to higher dimensions.}

\subsection{Transformations}

The transformations we consider are natural stochastic extensions of the transformation defined in \cite{bulking1,bulking2} for the  classical deterministic CA. These transformations can be divided into two categories: \emph{trimming operations} which allow to trim unwanted parts off the dynamics, and \emph{rescaling transformations} which augment the set of states and/or the neighborhoods.

\subsubsection{Trimming operations} 

They are based on three ingredients: 1) renaming states; 2) restricting to a stable subset of states; and 3) merging  compatible states. These ingredients are synthetized into two definitions (state renaming is implicit in both definitions).

\begin{definition}
Let $\CAA=(Q,R,V,V',f)$ be a stochastic CA. 
\begin{itemize}
\item if $i:Q'\rightarrow Q$ is an injective function such that $Y=(i(Q'))^\ZZ$ is $F$-stable (i.e. $F(Y,R^\ZZ)\subseteq Y$) then the \emph{$i$-restriction} of $\CAA$ is the stochastic CA:
  \[\rest{i}{\CAA}=(Q',R,V,V',\rest{i}{f})\]
  where $\rest{i}{f}$ is the local function associated with the explicit global function $\rest{i}{F}$ such that, $\forall c\in I$, $\forall s\in R^\ZZ$, $\rest{i}{F}(c,s) = I^{-1}\circ F(I(c),s)$ where $I:Q'^\ZZ\rightarrow Q^\ZZ$ denotes the cell-by-cell extension of $i$;
\item if $\pi:Q\rightarrow Q'$ is surjective and $F$-compatible (\textit{s.t.} ${\Pi\circ F(c,s) = \Pi\circ F(c',s)}$ for all $s$ and all $c,c'$ such that $\Pi(c)=\Pi(c')$, where $\Pi:Q^\ZZ\rightarrow Q'^\ZZ$ is the cell-by-cell extension of $\pi$), then the \emph{$\pi$-projection} of $\CAA$ is the stochastic~CA:
  \[\proj{\pi}{\CAA}=(Q',R,V,V',\proj{\pi}{f})\]
  where $\proj{\pi}{f}$ is the local function associated with the explicit global function $\proj{\pi}{F}$ such that $\proj{\pi}{F}(c',s) = \Pi\circ F(c,s)$ where $c$ is any configuration in $\Pi^{-1}(c')$.
\end{itemize}
\end{definition}

If $i:Q'\rightarrow Q$ and $\pi:Q'\rightarrow Q''$ verify the required  stability and compatibility conditions, we denote by $\mix{\pi}{i}{\CAA}$ the $\pi$-projection of the $i$-restriction of \CAA{}.

\begin{definition}\label{def:locrel}
  Let $\CAA_1=(Q_1,R_1,V_1,V_1',f_1)$ and $\CAA_2=(Q_2,R_2,V_2,V_2',f_2)$ be two arbitrary stochastic CA. We define the following relations:
  \begin{itemize}
  \item $\CAA_1\SSUB\CAA_2$, $\CAA_1$ is a \emph{stochastic subautomaton} of $\CAA_2$, if there is some \mbox{$i$-restriction} of $\CAA_2$ such that $\STOC{F_1}=\STOC{\rest{i}{F_2}}$;
  \item $\CAA_1\SPROJ\CAA_2$, $\CAA_1$ is a \emph{stochastic factor} of $\CAA_2$, if there is some $\pi$-projection of $\CAA_2$ such that $\STOC{F_1}=\STOC{\proj{\pi}{F_2}}$;
    \DELETE{ \item $\CAA_1\NSUB\CAA_2$, $\CAA_1$ is a \emph{non-deterministic subautomaton} of $\CAA_2$, if there is some $i$-restriction of $\CAA_2$ such that $\NDET{F_1}=\NDET{\rest{i}{F_2}}$;
    \item $\CAA_1\NPROJ\CAA_2$, $\CAA_1$ is a \emph{non-deterministic factor} of $\CAA_2$, if there is some $\pi$-projection of $\CAA_2$ such that $\NDET{F_1}=\NDET{\proj{\pi}{F_2}}$.
    \end{itemize}
    Moreover, if $\CAA_1$ and $\CAA_2$ are deterministic, we also define
    \begin{itemize}
    \item $\CAA_1\DSUB\CAA_2$, $\CAA_1$ is a \emph{deterministic subautomaton} of $\CAA_2$, if there is some \mbox{$i$-restriction} of $\CAA_2$ such that $\DET{F_1}=\DET{\rest{i}{F_2}}$;
    \item $\CAA_1\DPROJ\CAA_2$, $\CAA_1$ is a \emph{deterministic factor} of $\CAA_2$, if there is some $\pi$-projection of $\CAA_2$ such that $\DET{F_1}=\DET{\proj{\pi}{F_2}}$.}
  \end{itemize}
  Similarly, we define $\NSUB$ and $\NPROJ$ (for non-deterministic global maps) and $\DSUB$ and $\DPROJ$ (for deterministic global maps).
  We also define the three relations $\DMIX$, $\NMIX$ and $\SMIX$ using projections of restrictions. For instance: $\CAA_1\SMIX\CAA_2$ if there are $i$ and $\pi$ such that $\STOC{F_1}=\STOC{\mix{\pi}{i}{F_2}}$.
\end{definition}

\subsubsection{Rescaling transformations.}

The transformations defined so far only allow to derive a finite number of CA from a given CA (up to renaming of the states) and thus induce only a finite number of dynamics. In particular, the size of the set of states, and the size of the neighborhood, can only decrease. Following the approach taken for classical deterministic CA, we now consider \emph{rescaling transformations}, which allow to increase the set of states, the neighborhood, etc\DELETE{ while preserving the dynamics, in fact preserving the complete information in the space-time diagram}. Rescaling transformations consist in: composing with a fixed translation, packing cells into fixed-size blocks, and iterating the rule a fixed number of times. Notice that since stochastic CA are composable, they are stable under rescaling operations, whereas \CFCA{} are not.

The \emph{translation} $\shift{z}$ (for $z\in\ZZ$) is the deterministic CA whose deterministic global function verifies: ${\forall c,\forall z', \DET{\shift{z}}(c)_{z'} = c_{z'+z}}$.

Given any finite set $S$ and any $m\geq 1$, we define the bijective \emph{packing map} ${\bloc{m}: S^\ZZ\rightarrow
  \bigl(S^m\bigr)^\ZZ}$ by ${\bloc{m}(c)_z = (c_{mz},c_{mz+1},\ldots,c_{mz+m-1})}$
for all $c$ and $z$.

\begin{definition}
  Let $\CAA=(Q,R,V,V',f)$ be any stochastic CA. Let $m,t\geq 1$ and $z\in\ZZ$. The \emph{rescaling} of \CAA{} with parameters $(m,t,z)$ is the stochastic CA $\grp{\CAA}{m,t,z} = \bigl(Q^m, (R^m)^t, V_+, V_+',\grp{f}{m,t,z}\bigr)$ whose explicit global function $\grp{F}{m,t,z}$ is defined by:
  \[\grp{F}{m,t,z}(c,s) = \bloc{m}\circ \shift{z}\circ F^t(\debloc{m}(c),\debloc{m}(s^1),\ldots,\debloc{m}(s^t))\]
  where $s^1,\ldots,s^t\in (R^m)^\ZZ$ are the $t$ components of $s$ (s.t. $s^i_j = (s_j)_i$), and $V_+, V_+'$ the modified neighbourhoods following $b_m$.
\end{definition}

\subsection{Simulation Pre-Orders}

We can now define the general simulation relations.

\begin{definition}
  For each local relation $\somerel$ among the nine relations of Definition~\ref{def:locrel}, we define the associated simulation relation $\simu$ by
  \[\CAA_1\simu\CAA_2 \Leftrightarrow \exists m_1,m_2,t_1,t_2,z_1,z_2, \grp{\CAA_1}{m_1,t_1,z_1}\,\somerel\,\grp{\CAA_2}{m_2,t_2,z_2}\]
We therefore define nine simulation relations $\ssimui$, $\ssimus$, $\ssimum$, $\nsimui$, $\nsimus$, $\nsimum$, $\dsimui$, $\dsimus$ and $\dsimum$, where the subscript denotes the kind of local relation used (\textbf{i}njection, \textbf{$\pmb\pi$\hspace{-.2mm}}rojection or \textbf{m}ixed) and the superscript denotes the kind of global functions which are compared (\textbf{S}tochastic, \textbf{N}on-deterministic or \textbf{D}eterministic).

\DELETE{  To simplify notations and help the reader, we denote by 'i' (like injection) the local relations using restriction, by '$\pi$'  those using projection, and by 'm' (like mixed) those using restriction of projection. Hence, we define the nine following simulation relations:
  \begin{center}
    \small
    \begin{tabular}{r|c|c|c}
      & injection & surjection & mixed\\
      \hline
      &&&\\
      deterministic & $\dsimui$ & $\dsimus$ & $\dsimum$\\
      &&&\\
      non-deterministic & $\nsimui$ & $\nsimus$ & $\nsimum$\\
      &&&\\
      stochastic & $\ssimui$ & $\ssimus$ & $\ssimum$\\
    \end{tabular}
  \end{center}}
\end{definition}

\begin{lemma}\label{lem:restproj}
A restriction (resp. projection) of a restriction (resp. projection) of some stochastic CA \CAA{} is a restriction (resp. projection) of \CAA{}. 
Moreover, any restriction of a projection of \CAA{} is the projection of some restriction of \CAA{}.
\end{lemma}
\begin{proof}
  This is a straightforward generalization of the corresponding result in the classical deterministic settings. A detailed proof for the deterministic case appears in Theorem~2.1 of \cite{bulking2}. All arguments given in the proof are easily adaptable to our setting.
\end{proof}

The lemma above implies that any sequence of admissible restrictions and projections can be expressed as the projection of some restriction. 

From Lemma~\ref{lem:restproj} it follows that all local relations defined are transitive and reflexive. Moreover, the deterministic relations $\DSUB$ and $\DPROJ$ are exactly the same as those defined in the classical setting of deterministic CA \cite{bulking2}.

\begin{fact}\label{fac:simupreorders}
  All simulation relations $\ssimui$, $\ssimus$, $\ssimum$, $\nsimui$, $\nsimus$, $\nsimum$, $\dsimui$, $\dsimus$, $\dsimum$  are pre-orders.
\end{fact}

\begin{proof}
  It is sufficient to verify that for any local comparison relation $\somerel$:
  \begin{enumerate}
  \item $\somerel$ is compatible with rescalings, i.e.
    \[\CAA_1\somerel\CAA_2\Rightarrow \grp{\CAA_1}{m,t,z}\somerel\grp{\CAA_2}{m,t,z}\]
  \item rescalings are commutative with respect to $\somerel$, i.e.
    \[\grp{\grp{\CAA_1}{m,t,k}}{m',t',z'} \somerel \grp{\grp{\CAA_1}{m',t',z'}}{m,t,z}\]
  \end{enumerate}
  Both properties are straightforward from the definitions. Then, the transitivity of any simulation relation follows from the transitivity of the corresponding local comparison relation $\somerel$.
\end{proof}

Each stochastic pre-order is a refinement of the corresponding non-deterministic pre-order as shown by the following fact (straightforward corollary of Fact~\ref{fac:stocndet}).

\begin{fact}\label{fac:ifssimuthennsimu}
  If $\CAA_1\ssimui\CAA_2$ then $\CAA_1\nsimui\CAA_2$. The same is true for pre-orders $\ssimus$, $\ssimum$ and the corresponding (non-)deterministic pre-orders.
\end{fact}

Note that for any simulation relation $\simu$, ${\CAA_1\simu\CAA_2}$ means that two global functions are equal where one is obtained by applying \emph{only} space-time-diagram-preserving rescaling transformations to $\CAA_1$ (the simulated CA) and the other is obtained by applying both rescaling transformations and trimming operations to $\CAA_2$ (the simulator).

\subsection{Classifications of Stochastic Cellular Automata}

Simulation pre-orders can be seen as a tool to classify the behaviors of CA  \cite{bulking1,bulking2}. They can be used to formalize in a more precise way the empirical classes defined historically through experimentations. We now give some results on the structure induced on stochastic CA by this classification.

\paragraph{Ideals.} Some classes of stochastic CA may only simulate CA of their own class. This is the case of the deterministic CA and also of the class of the \emph{noisy} CA which are the CA $F$ such that ${\NDET{F}(c)=Q^\ZZ}$ for all~$c$. 

\begin{fact}\label{fact:ideals}
  Let $\simu$ be any non-deterministic or stochastic pre-order. Let $\CAA_1$ and $\CAA_2$ be stochastic CA such that $\CAA_1\simu\CAA_2$. If $\CAA_2$ is deterministic (resp. noisy) then $\CAA_1$ is deterministic (resp. noisy).
\end{fact}
\begin{proof}
  By Fact~\ref{fac:ifssimuthennsimu} it is sufficient to prove this for non-deterministic simulations. The property that the explicit global function is deterministic or noisy (i.e. surjective on each configuration) is preserved by rescaling transformation. Hence it is sufficient to check that being deterministic or noisy is preserved by restriction and projection. This is straightforward for projection (because a projection is an onto map). Determinism is clearly preserved by restriction. Moreover, a noisy stochastic CA does not admit any non-trivial restriction because no subset of states is stable under iteration. Hence, the restriction of a noisy CA is necessarily itself (up to renaming of states) or the trivial CA with only one state. Both are noisy and the fact follows.
\end{proof}

\paragraph{Simulation of stochastic CA  by a \CFCA{}.} Even if some stochastic CA cannot be expressed as a \CFCA{} (because of potential local probabilistic correlations), each can be simulated by a particular \CFCA{}. 
 \label{sec:pCFCAsimul}

\begin{theorem}\label{thm:CFCAsimu}
  For any stochastic CA $\AUTO A=(Q,R,V,V',f_A)$ there is a \CFCA{} $\AUTO B$ such that $\CAA\ssimui\AUTO B$.
\end{theorem}

\begin{proof}
The idea is to simulate one step of \CAA{} by two steps of $\AUTO B$:
\begin{enumerate}
\item generate a random symbol locally and copy it to a component of states;
\item simulate a stochastic transition of \CAA{} reading states only and ignoring random symbols.
\end{enumerate}
Formally, let $\AUTO B=(Q_B,R,V,V',f_B)$ where $Q_B=Q\cup Q\times R$ and $f_B$ is any local function such that the associated explicit global function $F_B$ verifies:
\begin{enumerate}
\item for any ${c\in Q^\ZZ\subseteq Q_B^\ZZ}$ and any $s\in R^\ZZ$, ${\bigl(F_B(c,s)\bigr)_z = (c_z,s_z)}$
\item for any ${c\in (Q\times R)^\ZZ\subseteq Q_B^\ZZ}$ and any $s\in R^\ZZ$, ${F_B(c,s) = F_A(\pi_Q(c),\pi_R(c))}$ where $\pi_Q$ and $\pi_R$ are cell-by-cell projections on $Q$ and $R$ respectively.
\end{enumerate}
It is straightforward to check that ${\CAA\SSUB{\AUTO B}^2}$ with the restriction induced by the identity injection ${i : Q^\ZZ\rightarrow Q^\ZZ\subseteq Q_B^\ZZ}$.  
\end{proof}

Note that the restriction is essential in the above construction since the behavior is not specified (and no correct behavior can be specified) on configurations where states of type $Q$ and states of type $Q\times R$ are mixed. In particular it is \textbf{false} that the stochastic CA is the square of the \CFCA{}; it is a restriction of that.

Still, one could think that we might achieve a simpler simulation by taking $Q_B=Q\times R$ and doing the two steps simultaneously so that ${F_B(c,s)}$ would be the cell by cell product of ${F_A(\pi_Q(c),\pi_R(c))}$ and $s$. But this does not work: for such a $\AUTO B$ there is generally  no restriction  nor projection nor combination of both able to reproduce the stochastic global function of $\AUTO A$. Indeed, if some $c$ and $s_1,s_2$ are such that ${F_A(c,s_1)\neq F_A(c,s_2)}$ there is no valid way to define  a corresponding configuration for $c$ in $F_B$ because the $Q$-component of states in $F_B$ depends only on the previous deterministic configuration, not on the random configuration. Then, one might see this impossibility as an argument against our formalism of simulation. Of course, many extensions of our definitions might be considered to allow more simulations between stochastic CA. However, we think that the random component of the simulated CA should never be used to determine which deterministic configuration of the simulator CA corresponds to which deterministic configuration of the simulated CA. Doing so would be like predicting the noise of a system to prepare the state of another system. In particular, we do not see any reasonable formal setting where $F_B$ defined as above would be able to simulate $F_A$. $F_A$ and $F_B$ might look like two syntactical variants of essentially the same object, but, as stochastic dynamical systems, they are very different. For instance, not every configuration can be reached from any configuration in $F_B$ whereas $F_A$  could have this property (i.e. be a noisy stochastic CA).

We believe that a better understanding of the relationship between stochastic CA and \CFCA{} should go through the following questions: Is there a \CFCA{} in any equivalence class induced by the pre-order $\ssimui$? Is any stochastic CA $\ssimus$-simulated by some \CFCA{}?

\section{Universality}
\label{sec:univ}

The quest for universal CA is as old as the model itself. Intrinsic universality has also a long story as reported in \cite{OllingerUnivhistory}. Our formalism of simulation allows to extend the quest to stochastic cellular automata. 

Indeed, one of the main by-products of each simulation pre-order defined above is a notion of intrinsic universality. Formally, given some simulation pre-order $\simu$, a stochastic CA $\CAA$ is \emph{$\simu$-universal} if for any stochastic CA $\AUTO B$ we have ${\AUTO B\simu \CAA}$. When considering deterministic pre-orders, we recover the notions of universality already studied in the literature for classical deterministic CA \cite{Ollinger6states,OllingerRichard4states,bulking2}.

\subsection{Negative results}

When considering non-deterministic or stochastic global functions, the random symbols are hidden. Still, the choice of the set of random symbols plays an important role in the global functions we can possibly obtain. We denote by $\PF(n)$ the set of the prime factors of~$n$. By extension, for a stochastic CA \CAA{} with set of random symbols $R$, we denote by $\PF(\CAA)$ the set $\PF(|R|)$. We have the following result:

\begin{lemma}\label{lem:primefactors}
  Let $\AUTO A_1=(Q,R_1,V_1,V_1',f_1)$ and $\AUTO A_2=(Q,R_2,V_2,V_2',f_2)$ be two stochastic CA with same set of states. If they are not deterministic and ${\STOC{F_1}=\STOC{F_2}}$ then $\PF({\AUTO A}_1)\cap\PF({\AUTO A}_2)\neq\emptyset$.
\end{lemma}
\begin{proof} If $\CAA_1$ is not deterministic, then there must exist some configuration $c\in Q^\ZZ$ and two configurations $y\neq y'$ such that $\{y,y'\}\subseteq\NDET{F_1}(c)$. So there are two disjoint cylinders $\cyl{u}{z}\cap\cyl{u'}{z} = \emptyset$ with $y\in\cyl{u}{z}$ and $y'\in\cyl{u'}{z}$. Therefore ${0<\bigl(\STOC{F}(c)\bigr)(\cyl{u}{z})<1}$. Besides, by definition of $\STOC{F}$, we have 
\[\bigl(\STOC{F_1}(c)\bigr)(\cyl{u}{z}) = \nu_1(\anevent^1_{c,\cyl{u}{z}}) = \frac{p}{q}<1\]
for some relatively prime numbers $p$ and $q$ (recall that $\nu_1$ is the uniform measure over $R_1^\ZZ$ and that ${\anevent^1_{c,\cyl{u}{z}}=\{s\in R_1^\ZZ : F_1(c,s)\in\cyl{u}{z}\}}$). Moreover, ${\PF(q)\subseteq\PF(|R_1|)=\PF(\CAA_1)}$ since $\anevent^1_{c,\cyl{u}{z}}$ is a finite union of cylinders and since the $\nu_1$-measure of any cylinder is a rational of the form $\frac{a}{|R_1|^b}$ for some integers $a,b\geq 1$.
Now, by hypothesis, we have also
\[\bigl(\STOC{F_2}(c)\bigr)(\cyl{u}{z}) = \frac{p}{q}\]
and by a similar argument as above we deduce that $\PF(q)\subseteq\PF(\CAA_2)$. The lemma follows since $\PF(q)\neq\emptyset$ (because $\frac pq<1$).
\end{proof}

From Lemma \ref{lem:primefactors} it follows, surprisingly perhaps, that the random symbols of a stochastic CA limit its simulation power to stochastic CA that have compatible random symbols.
\begin{theorem}\label{thm:primefactors}
  Let $\simu$ be any stochastic simulation pre-order, and $\CAA_1$ and $\CAA_2$ two stochastic CA which are not deterministic. If $\CAA_1\simu\CAA_2$ then ${\PF({\AUTO A}_1)\cap\PF({\AUTO A}_2)\neq\emptyset}$.
\end{theorem}
\begin{proof} Trimming operations (restrictions and projections) do not modify the set of random symbols. Rescaling transformations modify the set of random symbols in the following way: $R\mapsto R^n$ for some integer $n$. Therefore such transformations preserve the set of prime factors $\PF(\CAA)$ of the considered CA \CAA. Moreover, rescaling transformations do not affect determinism: the rescaled version of a CA which is not deterministic cannot be deterministic. Hence, the relation $\CAA_1\simu\CAA_2$ implies an equality of stochastic global functions of two CA which have the same prime factors as $\CAA_1$ and $\CAA_2$, one of which is not deterministic. Therefore none of them is deterministic and the theorem follows from lemma~\ref{lem:primefactors}.
\end{proof}
\DELETE{\begin{proof}(Sketch) First $\PF{\cdot}$ is invariant under any rescaling transformation. Then, if the distribution matches, there must be coupling of the random strings of some  finite lenght by Theorem~\ref{thm:coupling}. The existence of such a coupling implies that $|R_1|$ and $|R_2|$ must have a common prime factor. 
\end{proof}}
\DELETE{Given any stochastic pre-order $\simu$, a $\simu$-universal CA cannot be deterministic because determinism is preserved by simulation (Fact~\ref{fact:ideals}). Moreover, for any prime~$p$ there is some $\CAA_p$ with $\PF(\CAA_p)=\{p\}$. Hence Theorem~\ref{thm:primefactors},}

The consequence in terms of universality is immediate and breaks our hopes for a stochastic universality construction. 
\begin{corollary} \label{cor:no:stoc:universal}
  Let $\simu$ be any stochastic simulation pre-order. There is no \mbox{$\simu$-universal} stochastic CA.
\end{corollary}

\subsection{Positive results}

Still, the negative result of Corollary \ref{cor:no:stoc:universal} leaves open the possibility of partial universality constructions. We will now describe how to construct a stochastic CA which is $\nsimui$-universal (hence also $\nsimum$-universal; however note that the existence of a  $\nsimus$- or even of a $\dsimus$-universal is still open), and then draw the consequences.\\
Since we are not concerned with size optimization, we will use simple construction techniques using parallel Turing heads and table lookup as described for classical deterministic CA in \cite{OllingerUnivhistory}. More precisely, we construct a stochastic CA $\AUTO U=(Q_U,R_u,V_U,V_U',f_U)$ able to $\nsimui$-simulate any stochastic CA $\AUTO A=(Q,R,V,V',f)$ with no rescaling transformation on $\AUTO A$ and no shift in the rescaling of $\AUTO U$. Therefore each cell of $\AUTO A$ will be simulated by a block of $m$ cells of $\AUTO U$ and each step of $\CAA$ will be simulated by $t$ steps of $\AUTO U$ ($t$ and $m$ depend on \CAA{} and are to be determined later).\\
The blocks of $m$ cells have the following structure (the restriction in the pre-order handles the trimming of any invalid block):
\begin{center}\scriptsize\sf
    \begin{tabular}{|c|c|c|c|c|c|}
      \hline
      SYNC & transition table & $Q$-state & $R$-symbol & $Q$-states of neighbors & $R$-symbols of neighbors \\
      \hline
    \end{tabular}
\end{center}
where each part uses a fixed alphabet (independent of $Q$ and $R$) and only the width of each part may depend on \CAA{}\DELETE{ (note that the transition table encodes in particular the size of $Q,R,V,V'$)}. To each such block is attached a Turing head which will repeat cyclically a sequence of $4$ steps (sub-routines) described below. On a complete configuration made of such blocks there will be infinitely many such heads (one per block) executing these steps in parallel. Execution is synchronized at the end of each step (\textsf{SYNC} part) and such that two Turing heads never collide. More precisely, for some steps (2 and 4) the moves of all heads are rigorously identical (hence synchronous and without head collision). For some other steps (1 and 3), the sequence of moves of each head depends on the content of their corresponding block but these steps are always such that the head does not go outside the block (hence no risk of head collision) and they are synchronized at the end by the \textsf{SYNC} part which implements a small time countdown initialized to the maximum time needed to complete the step in the worst case. The parts holding $R$-symbols are initially empty (uniformly equal to some symbol) for each block. The $4$ steps are as follows:
\begin{enumerate}
\item generate \DELETE{surjectively }a string representing a random $R$-symbol in the \textsf{$R$-symbol} part using (possibly several) random $R_U$-symbols present in that part of the block;
\item copy the \textsf{$R$-symbol} part to the appropriate position in the \textsf{$R$-symbols of neighbors} part of each neighboring block. Do the same for \textsf{$Q$-state};
\item using information about $Q$-states and $R$-symbols in the block, find the corresponding entry in the transition table and update the \textsf{$Q$-states} part of the block accordingly;
\item clean \textsf{$R$-symbol} and \textsf{$R$-symbols of neighbors} parts (i.e. write some uniform symbol everywhere).
\end{enumerate}

This construction scheme is very similar to the one used for classical deterministic CA but two points are important in our context:
\begin{itemize}
\item step $4$ is here to ensure that each configuration of \CAA{} has a canonical corresponding configuration of $\AUTO U$ made of blocks where the parts holding $R$-symbols is clean (i.e., step~4 is required for the existence of the injection~$i$);
\item depending on the way we generate strings representing an $R$-symbols from strings of $R_U$-symbols in step $1$, we will obtain or not a uniform distribution over $R$ (recall Theorem~\ref{thm:primefactors}).
\end{itemize}

In the general case, we can always fix (by the means of the injection~$i$) a  width large enough for the \textsf{$R$-symbols} parts containing  so that all $R$-symbols can be obtained (but with possibly different probabilities). We therefore obtain a universality result for non-deterministic simulations.

\begin{theorem}\label{thm:nondet:univ}
  Let $\simu$ be either $\nsimui$ or $\nsimum$. There exists a $\simu$-universal CA.
\end{theorem}

Note that this  $\simu$-universal CA is a \CFCA, and we obtain thus a stronger version of the simulation mentioned in Section~\ref{sec:pCFCAsimul} page~\pageref{sec:pCFCAsimul}.

Now, if we are in a case where ${\PF(\CAA)\subseteq\PF(U)}$ then it is possible to choose a generation process in step $1$ such that each $R$-symbol is generated with the same probability. We therefore obtain an optimal partial universality construction for stochastic simulations.

\begin{theorem}\label{thm:stoc:univ}
  Let $\simu$ be either $\ssimui$ or $\ssimum$. For any finite set $P$ of prime numbers, there is a stochastic CA $\AUTO U_P$ such that for any stochastic CA \CAA: ${\PF(\CAA)\subseteq P\Rightarrow \CAA\simu\AUTO U_P.}$
Moreover,  $\AUTO U_P$ is a \CFCA.
\end{theorem}

\section{Recap of Some Decidability Results}

\newcommand\Ude{\textit{\textcolor{red}{Undecidable}}}
\newcommand\Dec{\textit{\textcolor{black}{Decidable}}}

\begin{center}
\begin{tabular}{ccccc}
\bfseries Problem  & \bfseries General & \bfseries Correlation-Free & \bfseries 1D & \bfseries Correlation-free 1D\\ \toprule
$F$ deterministic & \Dec &  \Dec &  \Dec &  \Dec \\ 
 \midrule
$\STOC{F}=\STOC{G}$ & \Ude & \Dec & \Dec & \Dec \\
 \midrule
 $F$ noisy & \Ude & \Dec & \Dec & \Dec \\
 \midrule
 $\STOC{F}^2=\STOC{G}^2$ & \Ude & \Ude & \Dec & \Dec \\
 \midrule
 $F^2$ noisy & \Ude & \Ude & \Dec & \Dec \\
 \midrule
\PPT{}
& \Ude & ? &  \Ude & \Dec \\
\bottomrule
\end{tabular}
\end{center}

\section{Open Problems}
\label{sec:open}

Intrinsic simulation has been proven to be a powerful tool to hierarchize behaviors in the deterministic world. In particular, the notion of universal CA allows to formalize the concept of ``most complex'' CA as the ones concentrating ``all the possible behaviors'' within a given class \cite{bulking1,bulking2}.
The formalism and the notion of intrinsic simulation developed here for stochastic CA, enable us to export this classification tool to the stochastic world. In particular, it would be interesting to see whether our partial universality construction relates to experimentally observed classes, as in \cite{RST}.\\
It would also be interesting to extend these notions of intrinsic simulation between stochastic CA to noisy Quantum Cellular Automata, as this could be of use for quantum simulation.\\
At the theoretical level, and amongst all the concrete questions raised by this article, the following ones are of particular interest:
\begin{itemize}
\item 
Is there for any stochastic \CAA{}, a \CFCA{} $\AUTO B$ which is  ${\ssimui}$-equivalent to ${\AUTO A}$?
\item 
Is there for any stochastic \CAA{}, a \CFCA{} $\AUTO B$ such that ${\CAA\ssimus\AUTO B}$?
\item 
Are there $\nsimus$-universal cellular automata?
\item  
Are universal CA the same for pre-order $\nsimui$ and $\nsimus$?
\item 
Is \PPT{} undecidable for CFCA in dimension $2$ and higher?
\end{itemize}

Our setting can also be generalized by taking any Bernouilli measure on the $R$-component (instead of the uniform measure). We believe that positive and negative results about universality essentially still hold but under a different form.

More generally, as far as we know, there is no characterization of the probability distributions over the configurations that correspond to images of cellular automata: deterministic automata starting from a random initial configuration, nor stochastic cellular automata starting from fixed or random distribution. In particular, we failed in our attempts to obtain an ``Hedlund-like'' characterization \cite{Hedlund} of the stochastic maps corresponding to stochastic cellular automata  (recall that there are \emph{constant} stochastic maps which do not correspond to any stochastic cellular automaton). One possible direction might be to explore extensions of our framework  allowing arbitrary shift-invariant distributions for the $R$-configuration; a characterization of this extension would however be still unsatisfying since many  shift-invariant distributions are highly non-local and are thus only remotely related to cellular automata.

\bibliographystyle{fundam}
\bibliography{uspca}

\end{document}